\documentclass[a4paper,onecolumn,11pt,accepted=2025-10-01]{quantumarticle}
\pdfoutput=1

\usepackage{graphicx} 
\usepackage{mathtools}
\usepackage[utf8]{inputenc}
\usepackage[english]{babel}
\usepackage[T1]{fontenc}
\usepackage{amsmath,amssymb,amsthm,amscd}
\usepackage{hyperref}
\usepackage[utf8]{inputenc}
\usepackage[english]{babel}
\usepackage[T1]{fontenc}
\usepackage{amsmath}
\usepackage{doi}

\usepackage{tikz}
\usepackage{lipsum}

\newtheorem{theorem}{Theorem}
\newtheorem{lem}[theorem]{Lemma}
\newtheorem{prop}[theorem]{Proposition}
\newtheorem{cor}[theorem]{Corollary}
\newtheorem*{prop*}{Proposition}

\usepackage{csquotes}

\usepackage{mathrsfs}
\usepackage[active]{srcltx}
\usepackage{pgf,tikz}

\usepackage{enumerate}

\usepackage{braket}
\usetikzlibrary{arrows}
\usepackage{dsfont}
\usepackage{paralist}
\usepackage{comment}
\usepackage{enumerate}

 \usepackage{bm}

\definecolor{navyblue}{rgb}{0.0, 0.0, 0.5}

\DeclareMathOperator{\suc}{Succ}

\newcommand*{\bid}{\mathbf{1}}

\newcommand*{\cC}{\mathcal{C}}

\newcommand*{\cE}{\mathcal{E}}

\newcommand*{\cG}{\mathcal{G}}
\newcommand*{\cH}{\mathcal{H}}
\newcommand*{\cI}{\mathcal{I}}

\newcommand*{\dI}{\mathbb{I}}

\newcommand*{\cN}{\mathcal{N}}

\newcommand*{\cK}{\mathcal{K}}
\newcommand*{\cM}{\mathcal{M}}
\newcommand*{\cL}{\mathcal{L}}

\newcommand*{\cP}{\mathcal{P}}

\newcommand*{\cS}{\mathcal{S}}
\newcommand*{\fS}{\mathfrak{S}}

\newcommand*{\cU}{\mathcal{U}}
\newcommand*{\cV}{\mathcal{V}}

\newcommand*{\cX}{\mathcal{X}}

\newcommand*{\eps}{\varepsilon}

\newcommand*{\id}{\mathrm{id}}
\newcommand*{\tr}[1]{\mathrm{Tr}\left[#1\right]}
\newcommand*{\ptr}[2]{\mathrm{Tr}_{#1}\left[#2\right]}

\newcommand*{\proj}[1]{|#1\rangle\!\langle #1|}

\newcommand*{\mge}{\succcurlyeq}
\newcommand*{\mle}{\preccurlyeq}

\newcommand*{\pr}[1]{\mathbb{P}\left[#1 \right]}

\newcommand*{\ex}[1]{\mathbb{E}\left[#1 \right]}
\newcommand*{\exs}[2]{\mathbb{E}_{#1}\left[#2 \right]}

\usepackage[backend=bibtex, maxnames =5]{biblatex}
\begin{document}

\title{Quantum channel coding:
 Approximation algorithms and strong converse exponents}

\author{Aadil Oufkir}
\email{aadil.oufkir@gmail.com}
\affiliation{The UM6P Vanguard Center, Mohammed VI Polytechnic University, Rocade Rabat-Sal\'e,
Technopolis, Morocco}
\affiliation{Institute for Quantum Information,
  RWTH Aachen University,
  Aachen, Germany}
\orcid{0000-0002-8594-1488}
\author{Mario Berta}
\orcid{0000-0002-0428-3429}
\affiliation{Institute for Quantum Information,
  RWTH Aachen University,
  Aachen, Germany}
\maketitle

\begin{abstract}
   We study relaxations of entanglement-assisted quantum channel coding and establish that non-signaling assistance and a natural semi-definite programming relaxation\,---\,termed meta-converse\,---\,are equivalent in terms of success probabilities. We then present a rounding procedure that transforms any non-signaling-assisted strategy into an entanglement-assisted one and prove an approximation ratio of $(1 - e^{-1})$ in success probabilities for the special case of measurement channels. For fully quantum channels, we give a weaker (dimension dependent) approximation ratio, that is nevertheless still tight to characterize the strong converse exponent of entanglement-assisted channel coding [Li and Yao, IEEE Tran.~Inf.~Theory (2024)]. Our derivations leverage ideas from position-based coding, quantum decoupling theorems, the matrix Chernoff inequality, and input flattening techniques.
\end{abstract}

\section{Introduction}

\subsection{Motivation}

Channel coding lies at the heart of quantum information theory, focusing on the reliable transmission of information over noisy quantum channels. The classical understanding of this field is shaped by asymptotic capacity theorems for classical~\cite{shannon1948mathematical}, classical-quantum~\cite{Holevo1998Jan,Schumacher1997Jul} and quantum channels~\cite{Lloyd97,Shor02,Devetak05,Bennett1999Oct,Bennett2002Oct}, which describe the maximum achievable communication rates over many channel uses under different types of assistance (such as, e.g., entanglement or feedback). More recent work has refined some of this understanding by going beyond first order asymptotic limits and exploring, e.g., the small deviation regime (for classical \cite{Str62,Hayashi09,polyanskiy2010channel} and quantum channels \cite{Tomamichel2015Aug,Datta2016Jun}), or the large deviation regime with error exponents (for classical \cite{Fano1961Nov,Shannon1967Jan,Haroutunian2008Feb} and quantum channels \cite{Dalai2013Sep, Beigi2023Oct,Li2024Jul,Renes2024Jul,Oufkir24}) as well as strong converse exponents (for classical \cite{Arimoto1973May,imre} and quantum channels \cite{Winter2002Aug,Ogawa2002Aug,Gupta2015Mar,Li2023Nov}). These refinements highlight the trade-offs between rate, reliability, and resources that arise in practical communication scenarios. An even more recent development in channel coding is the exploration of the true one-shot setting in a tight manner (see, e.g., \cite{polyanskiy2010channel,Buscemi2010Mar,Wang12,Matthews2014Sep,Leung2015Jun,barman2017algorithmic,Fawzi19,Wang2019Feb,Berta2022Jul} and references therein). This departs from the traditional assumption of infinite channel uses \cite{Tomamichel2016}. Instead, it focuses on the performance limits when only a single or finite number of uses are available. This approach opens the door to studying quantum communication in even more realistic, constrained scenarios, providing new insights into achievable rates and error behaviors.

In our work, we explore quantum channel coding in the one-shot framework from an algorithmic perspective and aim to provide a broader understanding of quantum channel coding, both in the asymptotic and non-asymptotic regimes. In classical channel coding, this idea was first introduced by \cite{barman2017algorithmic},   which provided a simple and efficient approximation algorithm that returns a code achieving a $(1 - \frac{1}{e})$ approximation of the maximum success probability attainable for a given classical channel. Notably, the approximation algorithm corresponds to coding over the same channel with more powerful assistance called non-signaling correlations. Moreover, in the classical setting non-signaling strategies are shown by \cite{matthews2012linear} to correspond to  a natural linear programming relaxation\,---\,the so-called meta-converse.\footnote{The meta-converse is a versatile linear programming bound that implies many converse results for channel coding \cite{polyanskiy2010channel}. The derivation and the formulation of the meta-converse use the hypothesis testing framework: the channel at hand is tested against an arbitrary constant channel.} In the same spirit, using non-signaling correlations to design approximation algorithms of the success probability has been applied to coding over classical-quantum channels \cite{Fawzi19},  multiple-access channels \cite{fawzi2023multiple}, and broadcast channels \cite{fawzi2023broadcast}.

Although not discussed in \cite{barman2017algorithmic}, it is straightforward to see that the multiplicative constant approximation $(1 - \frac{1}{e})$ is already sufficient to deduce that the plain (unassisted) strong converse exponent of a channel\,---\,i.e., the rate at which the success probability approaches $0$ when the transmission rate is strictly above the channel capacity\,---\,is the same as when non-signaling correlations are allowed. This allows for a connection between the strong converse exponent and the  exponent of  composite hypothesis testing, a problem widely studied and understood in the literature (see, e.g., \cite{Hayashi2016Oct}). This method can be seen as an alternative way of reproving the strong converse exponent, a result already established by \cite{Arimoto1973May,imre} (see also \cite{polyanskiy2010arimoto}). The advantage of this proof strategy is to discern between unassisted (or shared-randomness) and non-signaling assisted strategies. As we discuss later, we apply this method in the quantum setting as well.

In \cite{Fawzi19}, the authors generalized the approximation results of \cite{barman2017algorithmic} to the classical-quantum setting. In that work, two approximations of the success probabilities are proven: one is multiplicative, depending logarithmically on the size of the channel, and the other is additive, which is only non-trivial when the sizes of the codes differ. It turns out that the multiplicative approximation, while not as tight as the classical $(1 - \frac{1}{e})$ one proven in \cite{barman2017algorithmic},  is still sufficient to show that unassisted and non-signaling assisted strategies have the same strong converse exponent. This reduces the problem of finding the strong converse exponent to the case where non-signaling correlations are allowed. Perhaps surprisingly, coding with non-signaling correlations does not exactly correspond to the semi-definite programming meta-converse bound (or composite hypothesis testing) for classical-quantum channels \cite{Matthews2014Sep,Wang2019Feb}, unlike in the classical setting \cite{matthews2012linear}. Nevertheless, it can be shown that they have the same strong converse exponents even for quantum channels.

Motivated by these approximation results and their powerful implications for large deviation refinements, we pose the same questions in the quantum setting: 
\begin{center}
    \emph{Are similar one-shot approximation algorithms possible for quantum channel coding?}
    \\ \emph{If so, are they sufficient to establish large deviation refinements for quantum channels?}
\end{center}
In this work, we consider the problem of channel coding over quantum-classical and quantum channels with entanglement assistance in the one-shot setting. From an algorithmic perspective, as discussed in \cite{barman2017algorithmic,Fawzi19,fawzi2023multiple,fawzi2023broadcast,berta2024optimality}, our aim is to design efficient approximation algorithms that, given a complete description of the noisy channel, return near-optimal codes that maximize the success probability for a fixed number of messages to be transmitted.


\subsection{Overview of findings}

Following the works of \cite{polyanskiy2010channel,matthews2012linear, Matthews2014Sep,barman2017algorithmic,Fawzi19}, we consider a natural SDP relaxation for channel coding over quantum channels known as the meta-converse. This relaxation is closely related to coding with non-signaling assistance \cite{Wang2019Feb}. We then investigate the relation between the meta-converse,  non-signaling and entanglement assistance in the one-shot setting with a focus on the success probability. Moreover, we apply our one-shot findings along with known results for the hypothesis testing problem to deduce optimal strong converse exponents.

\paragraph{Non-signaling assistance and meta-converse.} The meta-converse (MC) reflects some constraints that any coding scheme should satisfy.  The meta-converse success probability is closely related to that of coding when the sender and receiver share non-signaling (NS) correlations. For coding over a quantum channel $\cN$ with a fixed number of messages $M$, the non-signaling success probability, denoted $\suc^{\rm{NS}}(\cN, M)$, is the solution to the following SDP program \cite{Matthews2014Sep,Wang2019Feb}:
\begin{equation}\label{eq-intro:NS-program}
    \begin{split}
\suc^{\rm{NS}}(\mathcal{N}, M)\!=\! 
\sup_{\rho_R, \Lambda_{RB}} \quad 
& \frac{1}{M} \tr{ \Lambda_{RB} \cdot (J_{\mathcal{N}})_{RB} }
\\
\text{subject to} \quad 
& \rho_R \in \mathcal{S}(R), \\
& \Lambda_B = \mathbb{I}_B, \\
& 0 \mle \Lambda_{RB} \mle M \rho_R \otimes \mathbb{I}_B.
\end{split}
\end{equation}

The MC success probability, denoted $\suc^{\rm{MC}}(\cN, M)$, is the solution to a similar program to \eqref{eq-intro:NS-program},  except that it has the constraint $\Lambda_B \mle  \mathbb{I}_B$ instead of $\Lambda_B = \mathbb{I}_B$. Classically, both programs are equal \cite{matthews2012linear} and thus non-signaling correlations provide an operational interpretation of the meta-converse introduced by \cite{polyanskiy2010channel}. However, in the quantum setting (even for classical-quantum channels), these programs are in general not equal. Our first result is a rounding inequality between the MC and NS success probabilities.

\begin{prop}\label{prop-intro:MC-NS}
    Let $\cN$ be a quantum channel and $M\ge 1$. We have that 
    \begin{align}
        \suc^{\rm{MC}}(\cN,M)&\ge  \suc^{\rm{NS}}(\cN,M) 
        \ge \left(1-\frac{1}{M}\right)\cdot   \suc^{\rm{MC}}(\cN,M).
    \end{align}
\end{prop}

This proposition shows that, as the number of messages increases, the optimal success probabilities for MC and NS strategies become increasingly similar. Since the approximation is multiplicative in terms of success probabilities, this rounding has a direct application on the strong converse exponents. In words, the strong converse exponent is the rate at which the success probability approaches $0$ when the transmission rate is strictly above the channel capacity. From Proposition \ref{prop-intro:MC-NS}, we deduce that  characterizing  the strong converse exponent of non-signaling strategies reduces to analyzing the meta-converse. Moreover, the meta-converse admits inherently a formulation in terms of the hypothesis testing problem. The strong converse exponent of the (composite) hypothesis testing problem can be derived from known results in the literature. More precisely, we can show using Proposition \ref{prop-intro:MC-NS} and  the results of \cite{Gupta2015Mar,Hayashi2016Oct}:

\begin{cor}\label{cor-intro:SCE-MC-NS}
   Let  $\cN$ be a  quantum channel and $r\ge 0$. We have that
    \begin{align}\label{intro:NS-exponent}
        \lim_{n\rightarrow \infty} -\frac{1}{n}\log\suc^{\rm{NS}}(\cN^{\otimes n},e^{nr}) 
        &=  \lim_{n\rightarrow \infty} -\frac{1}{n}\log\suc^{\rm{MC}}(\cN^{\otimes n},e^{nr})
        \\&= \sup_{\alpha\ge 0}  \frac{\alpha}{1+\alpha}\left( r- \widetilde{\cI}_{1+\alpha}(\cN)\right),
    \end{align}
    where $\widetilde{\cI}_{\alpha}(\cN)$ denotes the sandwiched mutual information of the channel $\cN$ and order $\alpha$, formally defined in \eqref{def:sand-mutual-info}. 
\end{cor}

In the following we show that the natural  NS relaxation \eqref{eq-intro:NS-program} provides an approximation of the entanglement-assisted success probability. We observe that because of Proposition \ref{prop-intro:MC-NS}, similar results can be phrased in terms of MC instead of NS.

\paragraph{Quantum-classical setting.} We begin our approximation results by the quantum-classical setting as the classical and classical-quantum settings were studied by \cite{barman2017algorithmic}  and \cite{Fawzi19} respectively. Moreover, the entanglement-assisted capacity of quantum-classical channels differs from that of the shared randomness and unassisted case \cite{Bennett1999Oct,Bennett2002Oct,Holevo2012Mar}. For this reason, we compare non-signaling assistance directly with   entanglement assistance. It turns out that in this setting, we are able to obtain nice approximation results that generalize exactly the optimal classical  approximation results of \cite{barman2017algorithmic} and relate the EA success probability, denoted $\suc^{\rm{EA}}$, with its NS counterpart. Namely we obtain:

\begin{prop}\label{prop-intro:rounding-qc}
Let $\cN$ be a quantum-classical channel and  $M, M'\ge 1$. We have that 
  \begin{align}
       \suc^{\rm{EA}}(\cN,M')
       &\ge  \frac{M}{M'}\left(1-\left(1-\frac{1}{M}\right)^{M'}\right) \suc^{\rm{NS}}(\cN,M).
    \end{align}
    In particular, when $M'=M$, we obtain
     \begin{align}
      \suc^{\rm{EA}}(\cN,M)\ge   \left(1-\frac{1}{e}\right) \suc^{\rm{NS}}(\cN,M).
    \end{align}
\end{prop}

The proof of this proposition utilizes the position-based coding of \cite{Anshu2018Jun} along with a sequential decoder similar to \cite{Wilde2013Sep}. The main advantage of the quantum-classical setting is that the output system is classical, allowing the sequential decoder operators to be significantly simplified, which enables a tight evaluation of the approximation error. 

Specifically, the NS program of the success probability \eqref{eq-intro:NS-program} provides a quantum state $\rho_R$ and an observable $O_{RB} = \frac{1}{M} \rho_R^{-1/2} \Lambda_{RB}\rho_R^{-1/2}$ such that 
\begin{align}\label{eq-intro:ns-equalities-qc}
      \tr{O_{RB} \cdot \sigma_{RB}} &= \suc^{\rm{NS}}(\cN,M)   \quad \text{and} \quad \tr{O_{RB} \cdot \sigma_R \otimes \sigma_B} = \frac{1}{M},
\end{align}
where $\sigma_{RB} = \rho_R^{1/2} (J_{\cN})_{RB}\rho_R^{1/2} = \cN_{A\rightarrow B}(\phi_{RA})$ and $\phi_{RA}$ is a purification of $\rho_R$.  We take inspiration from the position-based coding of \cite{Anshu2018Jun} and encode the message  using its corresponding position within a shared entangled state. The equalities in Eq.~\eqref{eq-intro:ns-equalities-qc} will be used to decode the message. More precisely, 
$M'$ copies of $\phi_{RA}$ are shared between Alice (holding $A_1\dots A_{M'}$ systems) and Bob (holding $R_1\cdots R_{M'}$ systems). To send the message $m\in [M']$, Alice applies the channel $\cN_{A_m \rightarrow B}$ to the $m^{\rm{th}}$ copy of her part of the shared entanglement. To decode the message, Bob performs some measurements on systems $R_1\cdots R_{M'} B$. By Eq.~\eqref{eq-intro:ns-equalities-qc}, measuring the systems $R_mB$ with the observable $O_{R_mB}$ successfully detects it (outcome ‘$0$') with a success probability $\suc^{\rm{NS}}(\cN,M)$, while measuring the systems $R_lB$ with the observable $O_{R_lB}$ for $l\neq m$  falsely detects (outcome ‘$0$') the message with probability $\frac{1}{M}$. A natural strategy is to measure systems  $R_lB$ for $l=1, \dots ,M'$ sequentially and return the first detection (outcome ‘$0$'). To analyze this scheme, we use crucially the fact that the operators $\{O_{R_lB}\}_{l=1}^{M'}$ do commute when $B$ is a classical system. In general ($B$ can be quantum), the operators $\{O_{R_lB}\}_{l=1}^{M'}$ do not commute, and we are unable to prove an approximation with a multiplicative error using the sequential decoder.

A direct application of Proposition \ref{prop-intro:rounding-qc} is that EA and NS strategies achieve the same strong converse exponent for quantum-classical channels. This along with the characterization of the strong converse exponent of NS in Corollary \ref{cor-intro:SCE-MC-NS} gives the EA strong converse exponent in the quantum-classical setting.

\begin{cor}\label{cor-intro:QC-EA-exponent}
 Let $\cN$ be a  quantum-classical channel and $r\ge 0$. We have that
    \begin{align}\label{intro:QC-EA-exponent}
        \lim_{n\rightarrow \infty} -\frac{1}{n}\log\suc^{\rm{EA}}(\cN^{\otimes n},e^{nr}) 
        &=  \lim_{n\rightarrow \infty} -\frac{1}{n}\log\suc^{\rm{NS}}(\cN^{\otimes n},e^{nr})
        \\&= \sup_{\alpha\ge 0}  \frac{\alpha}{1+\alpha}\left( r- \widetilde{\cI}_{1+\alpha}(\cN)\right).
    \end{align}
\end{cor}

We will further generalize this result to the quantum setting in Corollary \ref{cor-intro:Q-EA-exponent}. However, we are unable to achieve the same approximation ratios from Proposition \ref{prop-intro:rounding-qc} in the fully quantum setting. The main challenge is that the analysis of the sequential decoder becomes cumbersome when the channel  output system is quantum and the sequential measurement operators do not commute.

\paragraph{Quantum setting - Additive error.}
In the fully quantum setting, we can still use the position-based coding from \cite{Anshu2018Jun}. However, the sequential decoder is no longer easy to analyze. To address the issue of non-commutativity of the operators, we can apply the well-known decoder and inequality of Hayashi-Nagaoka \cite{Hayashi2003Jun}, which lead to an approximation with an additive error. Note that  this decoder and inequality were used by \cite{Fawzi19} to prove a similar approximation result in the classical-quantum setting, where a non-signaling strategy is rounded to a shared-randomness one.

\begin{prop}\label{prop-intro:rounding-qq}
Let $\cN$ be quantum channel and $M, M'\ge 1$. 
We have the following inequality between the entanglement-assisted and non-signaling success probabilities 
\begin{align}\label{eq-intro:add-err-sqrt}
  \suc_{}^{\rm{EA}}(\cN,M')&\ge   \suc^{\rm{NS}}(\cN,M)  -5 \sqrt{\frac{M'}{M}}. 
\end{align}
\end{prop}

This proposition shows that NS correlations do not help to increase the EA capacity of a quantum channel \cite{Leung2015Jun}. 
 Moreover, 
 the additive error  in \eqref{eq-intro:add-err-sqrt}  implies that EA and NS have the same second-order asymptotics in the small deviation regime (which remains not fully characterized in the quantum setting \cite{Datta2016Jun}, even for NS strategies). However, this approximation error is insufficient  to conclude that EA and NS have the same error or strong converse exponents.

Even in the classical-quantum setting, it is unclear whether EA and NS should have the same error exponent. This is because NS (specifically, with activation) corresponds to the sphere-packing bound \cite{Dalai2013Sep,Cheng2019Jan,Oufkir24}, while achievability results (such as those for rounded strategies) so far seem to only imply the random coding bound at best \cite{Beigi2023Oct,Li2024Jul,Renes2024Jul}, which is only known to be optimal  above a critical rate \cite{Li2024Jul,Renes2024Jul}. Therefore, we focus on investigating whether EA and NS share the same strong converse exponent. As we have seen in the quantum-classical setting, a simple approach  to compare EA and NS  strong converse exponents is by designing an EA approximation algorithm to the NS success probability that achieves a (small) multiplicative error. 

\paragraph{Quantum setting - Multiplicative error.} Recall that the NS program of the success probability \eqref{eq-intro:NS-program} provides a quantum state $\rho_R$ and an observable $O_{RB} = \frac{1}{M} \rho_R^{-1/2} \Lambda_{RB}\rho_R^{-1/2}$ such that 
\begin{equation} \label{eq-intro:success}
      \tr{O_{RB} \cdot \cN_{A\rightarrow B}(\phi_{RA})} = \suc^{\rm{NS}}(\cN,M) 
\end{equation}
where  $\phi_{RA}$ is a purification of $\rho_R$. In the quantum-classical setting, we can use the  position-based coding \cite{Anshu2018Jun} and the sequential decoder \cite{Wilde2013Sep}. However in the quantum setting,  the operators $\{O_{R_lB}\}_{l=1}^M$ do not commute and the sequential decoder is hard to analyze. To circumvent this difficulty, we draw inspiration from \cite{Fawzi19} and attempt to use  the following measurement device:
\begin{align}
    \left\{\Xi_l = \frac{1}{Z} O_{R_lB} \right\}_{l=1}^M \; \text{where} \;\; Z = \left\| \sum_{l=1}^M O_{R_lB} \right\|_\infty\!.
\end{align}
The normalization $Z$  ensures that $\sum_{l=1}^M \Xi_l \mle \dI$\,---\,a condition necessary for valid measurements. Moreover,  this normalization factor will contribute to the  approximation error as follows
\begin{align}
      \suc^{\rm{EA}}(\cN,M) \ge \frac{1}{Z}   \suc^{\rm{NS}}(\cN,M). 
\end{align}

The authors of \cite{Fawzi19} employ the matrix Chernoff inequality \cite{Tropp2015May} to control a similar normalization factor. In the quantum setting, however, applying the matrix Chernoff inequality is insufficient to control $Z$  with high probability. The reason is that it is unclear how to utilize the constraint
\begin{equation}
     \ptr{R}{\rho_R \cdot O_{RB}} = \frac{1}{M}\ptr{R}{ \Lambda_{RB}} = \frac{1}{M}\dI_B;
\end{equation}
as it does not appear in this form in the standard matrix Chernoff inequality. To overcome this obstacle, we consider a simplified setting where $\rho_R = \frac{1}{|R|}\dI_R$ is the maximally mixed state. Under this assumption, the constraint
$\ptr{R}{\rho_R \cdot O_{RB}} = \frac{1}{M}\dI_B$ simplifies to
\begin{equation}
    \ptr{R}{O_{RB}} = \frac{|R|}{M}\dI_B.
\end{equation}
Furthermore, it can be expressed as
\begin{align}
    &\exs{U_R\sim \cL(R)}{U_RO_{RB} U_R^{\dagger }} 
    = \frac{1}{|R|}\dI_R\otimes \ptr{R}{O_{RB}}
    =\frac{1}{M}\dI_{RB},
\end{align}
\sloppy where $\cL(R)$ is any unitary $1$-design probability distribution on $\mathbb{U}(R)$. Consequently, the constraint $\ptr{R}{\rho_R \cdot O_{RB}} = \frac{1}{M}\dI_B$ implies 
\begin{equation}
    \left\|\exs{U_R\sim \cL(R)}{U_RO_{RB} U_R^{\dagger }} \right\|_{\infty} = \frac{1}{M},
\end{equation}
which is one of the key factors of the matrix Chernoff inequality. Conjugating the observable $O_{RB}$ by the unitary $U_R$ leads to the same success probability if the state $\cN_{A\rightarrow B}(\phi_{RA})$ is appropriately conjugated (see Eq.~\eqref{eq-intro:success}).  Since $\phi_{RA}$ is the maximally entangled state (recall that we assume $\rho_R$ is the maximally mixed state and $\phi_{RA}$ is a purification of it), we can conjugate $\cN_{A\rightarrow B}(\phi_{RA})$ with the unitary $U_R$ by conjugating $\phi_{RA}$ with $U_A^\top $ (a property of the maximally entangled state) which can be performed by Alice before applying the channel. Note that as opposed to the position-based   coding, we only need one shared copy of the maximally entangled state. This is crucial for bounding the dimension of the measurement operators as this latter appears in the error probability of the matrix Chernoff inequality. Applying a random coding with a unitary $1$-design formed with Pauli operators was used previously for the problem of entanglement-assisted channel coding (see e.g., \cite{Hsieh2008Jun,Li2024Jul}).\\

In order to handle the case of an arbitrary state $\rho_R$ we utilize the input flattening technique from \cite{Li2024Jul}. The idea is to decompose the input state as 
\begin{equation}
    \rho_R = \bigoplus_{j=1}^v p_j \cdot  \frac{1}{|{\cH_j}|} \dI_{\cH_j},
\end{equation}
where $v$ is the number of distinct eigenvalues of $\rho_R$, $\{p_j\}_{j=1}^v$ is a probability distribution, and $\{\cH_j\}_{j=1}^v$ are orthogonal subspaces of $R$. 
With this in hand, we can sample from the probability distribution $j\sim \{p_j\}_{j=1}^v$ and consider the problem as if we have $\rho_R = \frac{1}{|{\cH_j}|} \dI_{\cH_j}$, a maximally mixed state. Our approximation strategy builds on these ideas and chooses the shared entanglement based on shared-randomness samples from $\{p_j\}_{j=1}^v$. Random unitaries sampled from a unitary $1$-design are applied depending on the corresponding subspaces. Finally, a measurement device is formed depending on the shared-randomness and random unitaries. This EA strategy is described in more details in Section \ref{sec:qq-multiplicative} and achieves the following performance.

\begin{prop}\label{prop-intro:qq-rounding}
Let $\cN_{A\rightarrow B}$ be a quantum channel. 
Let $\rho_R \in \cS(R)$ be an optimal  quantum state maximizing the NS program \eqref{eq-intro:NS-program}. 
Let $v$ be the number of distinct eigenvalues of $\rho_R$. Assume that $\log|A|\ge e^2$.  
    We have 
    \begin{align}
        &\suc^{\rm{EA}}(\cN,M) 
        \ge \frac{1}{2v \log\left(2v e^{v} M (\dim A)^{v}\dim B\right)}\cdot\suc^{\rm{NS}}(\cN,M) .
    \end{align}
\end{prop}

In the finite block-length regime, we then find
\begin{align}
\suc^{\rm{EA}}(\cN^{\otimes n},e^{nr}) \ge\frac{1}{\operatorname{poly}(n)}\cdot\suc^{\rm{NS}}(\cN^{\otimes n},e^{nr}),
\end{align}
which also explains why this multiplicative error is not harmful is terms of strong converse exponent. To obtain the previous inequality, we use symmetry arguments and show that a near-optimal  quantum state maximizing the NS program \eqref{eq-intro:NS-program} is the de Finetti state \cite{christandl2009postselection} for which we have  $v_n  =  \operatorname{poly}(n)$ \cite{Hayashi2016Oct}. This allows us to deduce the EA strong converse exponent from the NS one (Corollary \ref{cor-intro:SCE-MC-NS}).

\begin{cor}\label{cor-intro:Q-EA-exponent}
 Let  $\cN$ be a  quantum channel and $r\ge 0$. We have that
    \begin{align}\label{intro:Q-EA-exponent}
        \lim_{n\rightarrow \infty} -\frac{1}{n}\log\suc^{\rm{EA}}(\cN^{\otimes n},e^{nr})
        &=  \lim_{n\rightarrow \infty} -\frac{1}{n}\log\suc^{\rm{NS}}(\cN^{\otimes n},e^{nr})
        \\&= \sup_{\alpha\ge 0}  \frac{\alpha}{1+\alpha}\left( r- \widetilde{\cI}_{1+\alpha}(\cN)\right).
    \end{align}
\end{cor}

This tight characterization of the strong converse exponent in the quantum setting was first established by \cite{Li2023Nov}. Moreover, \cite{Li2023Nov} was the first to introduce the flattening technique for this problem. While our proof also employs this technique, its structure is significantly different from that of \cite{Li2023Nov}. In particular, our proof relies on the relationship between EA, NS, and MC success probabilities, providing deeper insights into the connection between composite hypothesis testing and entanglement-assisted coding.


\subsection{Notation}

The set of integers between $1$ and $n$ is denoted by $[n]$.
Finite dimensional Hilbert spaces are denoted by $A,B, \dots$ and their corresponding  dimensions are written as  $|A|, |B|, \dots$. The set of unitary operators on the space $A$ is denoted by $\mathbb{U}(A)$. 
A positive semi-definite operator is denoted by $\rho \mge 0$. Moreover, $\rho\mge \sigma$ stands for $\rho-\sigma\mge 0$. For a Hermitian operator $\rho$ with spectral decomposition $\rho = \sum_i \lambda_i \Pi_i $, we denote by $\rho_+ = \sum_{i} \max\{\lambda_i,0\} \Pi_i $  its positive  part, and $\rho^{-1} = \sum_{i : \lambda_i\neq 0} \lambda_i^{-1} \Pi_i $ its Moore-Penrose generalized inverse.

Quantum states are positive semi-definite operators with trace $1$. The set of quantum states on a Hilbert space $A$ is denoted by $\cS(A)$. Let $\{\ket{i}\}_{i=1}^{|A|}$ be the canonical orthonormal basis of $A$. The maximally mixed state over $A$ is $\bid_A = \frac{1}{|A|}\dI_A$. The maximally entangled state over $A$ is
\begin{align}
\ket{\Psi}_{RA} =\frac{1}{\sqrt{|A|}} \ket{w}_{RA}  = \frac{1}{\sqrt{|A|}} \sum_{i=1}^{|A|} \ket{i}_{R}\otimes \ket{i}_A,
\end{align}
where $R\simeq A$. 

A quantum channel $\cN_{A\rightarrow B} $ is a completely-positive, trace-preserving map from $\cS(A)$ to $\cS(B)$. 
The Choi matrix of $\cN_{A\rightarrow B} $ is
\begin{align}
(J_{\cN})_{RB}= \cN_{A\rightarrow B}(\proj{w}_{RA}) 
&=   \sum_{i,j=1}^{|A|}  \ket{i}\bra{j}_{R} \otimes \cN( \ket{i}\bra{j})_{B}.
\end{align}
A quantum-classical or measurement channel $\cN_{A\rightarrow X}$ is a quantum channel where the output system $X$ is classical. 
\\A measurement device is a POVM (positive operator-valued measure), consisting of a set of positive semi-definite operators that sum to the identity operator $\dI$. Measuring a quantum state $\rho$ with a POVM $\cM = \{M_x\}_{x\in \cX}$ returns the outcome $x\in \cX$ with probability $\tr{\rho M_x}$. 

The sandwiched R\'enyi divergence \cite{MDSFT2013on,WWY2014strong} between $\rho \mge 0$ and $\sigma \mge 0$ of order $\alpha>1$ is defined as follows 
\begin{align}\label{def:sand-div}
    \widetilde{D}_\alpha\left(\rho \middle\| \sigma \right) = \frac{1}{\alpha-1}\log \tr{\left(\sigma^{\frac{1-\alpha}{2\alpha}}\rho \sigma^{\frac{1-\alpha}{2\alpha}}\right)^{\alpha}}
\end{align}
if $\rho \ll \sigma$ (i.e., the support of $\rho$ is included in the support of $\sigma$). Otherwise, we set $   \widetilde{D}_\alpha\left(\rho \middle\| \sigma \right) = +\infty$. It extends  continuously for $\alpha=1$ to the Umegaki relative entropy, given by  
\begin{equation}
    D\left(\rho \middle\| \sigma \right) = \tr{\rho (\log\rho-\log\sigma)}.
\end{equation}
The quantum hypothesis testing relative entropy is defined as 
\begin{align}
   D_H^{\eps}(\rho\| \sigma) = -\log \min \Big\{\tr{\sigma O} \Big| \,0\mle O\mle \dI, \tr{\rho O}\ge 1-\eps \Big\}. 
\end{align}

For a quantum channel $\cN_{A \rightarrow B}$, the channel's sandwiched R\'enyi mutual information of order $\alpha\ge 1$ is defined as \cite{Gupta2015Mar}
\begin{align} \label{def:sand-mutual-info}
\widetilde{I}_{\alpha}(\cN_{A \rightarrow B})
&=\sup_{\rho_R \in \cS(R)}\inf_{\sigma_B\in \cS(B)}\widetilde{D}_{\alpha}\left( \rho_R^{1/2}(J_\cN)_{RB} \rho_R^{1/2}\middle\|\rho_R \otimes  \sigma_B\right).
\end{align}

Let $\sigma$ be a Hermitian operator with spectral projections $\Pi_1, \dots, \Pi_v$, where  $v$  is the number of distinct eigenvalues of $\sigma$. The pinching channel associated with $\sigma$ is defined as:
\begin{equation}
    \cP_{\sigma} (\cdot)= \sum_{i=1}^v \Pi_i (\cdot) \Pi_i.
\end{equation}
The pinching inequality \cite{Hayashi2002Dec,Tomamichel2016} states that  for any positive semi-definite operator $\rho$,
\begin{equation}\label{eq:pinching-ineq}
   \rho\mle v\cdot\cP_{\sigma} (\rho). 
\end{equation}


\section{Quantum channel coding}\label{sec:quantum-coding}

In this section, we define the task of quantum channel coding using entanglement-assisted or non-signaling strategies. Then we study the relation between non-signaling assistance and the meta-converse.


\subsection{Entanglement-assisted channel coding}\label{sec:EA-coding}

We follow \cite{wilde2013book} and define an $(n,M)$ entanglement-assisted code as follows. Alice and Bob share a pure quantum state $\ket{\Psi}_{E_AE_B}$ where Alice holds $E_A$ and Bob holds $E_B$. Let $[M]$ be the set of messages to be transmitted. To send a message $m\in [M]$, Alice applies an encoding channel $\cE^m_{E_A\rightarrow A'^n}$ to her part of the shared entangled state $\proj{\Psi}_{E_AE_B}$ depending on the message $m$. Then the  global state is 
\begin{equation}
    \cE^m_{E_A\rightarrow A'^n}\left(\proj{\Psi}_{E_AE_B}\right). 
\end{equation}
Alice transmits $A'^n$ over $n$ iid copies of the channel $\cN^{}_{A'\rightarrow B}$, leading to the state 
\begin{equation}
    \cN^{\otimes n}_{A'^n\rightarrow B^n}\left( \cE^m_{E_A\rightarrow A'^n}\left(\proj{\Psi}_{E_AE_B}\right)\right). 
\end{equation}
Upon receiving systems $B^n$, Bob performs a POVM $\{\Xi^m_{B^nE_B}\}_{m=1}^M$ on the channel output $B^n$ and his part of the shared entanglement $E_B$ in order to decode the message $m$ that Alice sent. Bob returns the outcome of this measurement. Let
\begin{align}
\cC = \Big(\ket{\Psi}_{E_AE_B},  \{\cE^m_{E_A \rightarrow A'^n}\}_{m=1}^M, \{\Xi^m_{B^nE_B}\}_{m=1}^M\Big).
\end{align}
The probability of Bob correctly decoding Alice's message $m$ is 
\begin{align}
    \suc_{\cC}(\cN^{\otimes n}, M, m) 
    &= \tr{\Xi^m_{B^nE_B} \cN^{\otimes n}_{A'^n\rightarrow B^n}\left( \cE^m_{E_A\rightarrow A'^n}\left(\proj{\Psi}_{E_AE_B}\right)\right)}. 
\end{align}
The average success probability for the coding scheme $\cC$ is 
\begin{align}
    \suc_{\cC}(\cN^{\otimes n}, M)  &= \!\frac{1}{M}\!\sum_{m=1}^M \! \mathrm{Tr}\!\left[\Xi^m_{B^nE_B} \cN^{\otimes n}_{A'^n\rightarrow B^n}\!\left(\! \cE^m_{E_A\rightarrow A'^n}\!\left({\Psi}_{E_AE_B}\right)\!\right)\!\right]\!. 
\end{align}
The rate $r$ of communication is 
\begin{equation}
    r= \frac{1}{n}\log M,
\end{equation}
and the code $\cC$ has $\eps$ error if $\suc_{\cC}(\cN^{\otimes n}, M)\ge 1-\eps$. A rate $r$ of entanglement-assisted classical communication is achievable if there is an $(n, e^{n(r-\delta)})$ entanglement-assisted classical code achieving a success probability $1-\eps$ for all $\eps\in(0,1)$, $\delta>0$, and sufficiently large $n$. 
 The entanglement-assisted classical capacity $C^{\rm{EA}}(\cN)$ of a quantum channel $\cN$ is equal to the supremum of all achievable rates. 

Finally, the  one-shot entanglement-assisted success  probability is  given as the optimal average success probability over all EA $(1,M)$ codes
\begin{equation}\label{eq:def-EA-suc}
    \suc^{\rm{EA}}(\cN, M) = \sup_{\cC}    \suc_{\cC}(\cN, M). 
\end{equation}
The one-shot entanglement-assisted error probability can be defined as follows
\begin{equation}\label{eq:def-EA-err}
    \eps^{\rm{EA}}(\cN, M) = 1-\suc^{\rm{EA}}(\cN, M).
\end{equation}


\subsection{Non-signaling assisted channel coding}\label{sec:NS-coding}

The computation of the optimal entanglement-assisted (EA) success probability in \eqref{eq:def-EA-suc} is notoriously difficult, as it requires optimization over  shared entangled states, encoding channels, and decoding measurement devices, all of which can involve systems with unbounded dimension. A natural relaxation of this optimization is obtained by allowing more powerful correlations known as non-signaling~\cite{Matthews2014Sep}. Strategies using this assistance encompass entanglement-assisted strategies, though communication cannot be achieved using non-signaling correlations alone.

A non-signaling assisted code corresponds to a superchannel which is non-signaling (does not permit communication by itself) from the encoder to the decoder and vice-versa. More precisely, a non-signaling channel between Alice and Bob whose input systems are $A_i, B_i$ and output systems are $A_o, B_o$ has a  positive semi-definite Choi matrix $J_{A_iA_oB_iB_o}$ satisfying (e.g., \cite{Fang2019Sep}):
\begin{align}
    J_{A_iB_iB_o}  &= \bid_{A_i} \otimes J_{B_iB_o}, &&(A\nrightarrow B), 
    \\J_{A_iA_oB_i}  &= J_{A_iA_o}\otimes \bid_{B_i},  &&(B\nrightarrow A). 
\end{align}
Let $\cN\!\cS$ be the set of non-signaling superchannels. The optimal non-signaling success probability of a quantum-channel $\cN$ for sending $M$ messages is 
\begin{align}\label{eq:eps-NS-def}
     \suc^{\rm{NS}}(\cN, M) 
     &=\sup_{\Pi \in \cN\!\cS}  \frac{1}{M} \sum_{m=1}^M \bra{m}\Pi (\cN) (\proj{m}) \ket{m},
\end{align}
where $\Pi (\cN)$ is the effective channel. 
The non-signaling error probability of a quantum-channel $\cN$ for sending $M$ messages can be defined as follows
\begin{equation}\label{eq:def-NS-err}
    \eps^{\rm{NS}}(\cN, M) = 1-\suc^{\rm{NS}}(\cN, M).
\end{equation}
It is known that the success probability \eqref{eq:eps-NS-def} is the solution of the following SDP program \cite{Matthews2014Sep,Wang2019Feb}
\begin{align}
\suc^{\rm{NS}}(\cN, M) = \sup_{\rho_R, \Lambda_{RB}} \quad & \frac{1}{M} \tr{\Lambda_{RB} (J_{\cN})_{RB}} \notag \\
\text{subject to} \quad
& \rho_R \in \cS(R), \notag \\
& \Lambda_B = \mathbb{I}_B, \notag \\
& 0 \mle \Lambda_{RB} \mle M \rho_R \otimes \mathbb{I}_B, \label{NS-program-pre}
\\[6pt]
= \sup_{\rho_R, \Lambda_{RB}} \quad & \tr{\Lambda_{RB} (J_{\cN})_{RB}} \notag \\
\text{subject to} \quad
& \rho_R \in \cS(R), \notag \\
& \Lambda_B = \frac{1}{M} \mathbb{I}_B, \notag \\
& 0 \mle \Lambda_{RB} \mle \rho_R \otimes \mathbb{I}_B. \label{NS-program}
\end{align}

Unlike the entanglement-assisted success probability \eqref{eq:def-EA-suc}, the non-signaling success probability is computationally accessible. Furthermore, since entanglement-assisted strategies are non-signaling, we have the obvious inequality
\begin{equation}\label{eq:NS>EA}
    \suc^{\rm{NS}}(\cN, M)\ge \suc^{\rm{EA}}(\cN, M).
\end{equation}
One of the goals of this paper is to design approximation algorithms that achieve a reverse type inequality of \eqref{eq:NS>EA} with as small an (additive or multiplicative) error as possible.


\subsection{Non-signaling value and meta-converse}\label{sec:NS-MC}

In the previous section, we observed that non-signaling strategies are useful for approximating the entanglement-assisted success probability of coding over quantum channels. It turns out that the non-signaling success probability is closely related to the well-known meta-converse \cite{polyanskiy2010channel,matthews2012linear,Matthews2014Sep,Wang2019Feb}. In fact, they are even equal in the classical setting \cite{matthews2012linear}. 
The meta-converse error probability for one-shot coding over a quantum channel $\cN$ and with a communication size $M$ is \cite{Matthews2014Sep}
\begin{align}\label{eq:mc-hypothesis}
  & \eps_{}^{\rm{MC}}(\cN,M) =\inf_{} \Big\{ \eps\in[0,1] \;\Big|\; 
\sup_{\rho_{R} \in \cS(R)} \inf_{\sigma_B\in \cS(B)}\! D_H^{\eps}\!\left(\rho_R^{1/2}(J_\cN)_{RB}\rho_{R}^{1/2} \big\| \rho_R \otimes \sigma_B \right)\ge \log M \!\Big\}. 
\end{align}
It is known that the meta-converse success probability  is the solution of the following SDP program \cite{Matthews2014Sep,Wang2019Feb}
\begin{align}
    \suc^{\rm{MC}}(\cN, M) = \sup_{\rho_R, \Lambda_{RB}} \quad & \frac{1}{M} \tr{\Lambda_{RB} (J_{\cN})_{RB}} \notag \\
\text{subject to} \quad
& \rho_R \in \cS(R), \notag \\
& \Lambda_B \mle  \mathbb{I}_B, \notag \\
& 0 \mle \Lambda_{RB} \mle M \rho_R \otimes \mathbb{I}_B.\label{MC-program}
\end{align}
Observe that the difference between the meta-converse program \eqref{MC-program} and the non-signaling program \eqref{NS-program-pre} for the success probability lies in the constraint $\Lambda_B \mle  \mathbb{I}_B$ versus $\Lambda_B  =  \mathbb{I}_B$. This implies that $\suc^{\rm{MC}}(\cN, M)  \ge \suc^{\rm{NS}}(\cN, M)$,  since the meta-converse allows for a looser constraint on the measurement operator $\Lambda_B$. The authors of \cite{Wang2019Feb} show that the meta-converse corresponds to non-signaling correlations assisted by $1$-bit of perfect communication (known as activation). In the following we provide a reverse inequality between the meta-converse and non-signaling success probabilities, achieving this  with a multiplicative error and without the need for  activation.

\begin{prop}\label{prop:NSvsMC}
Let $\cN$ be a quantum channel and $M\ge 1$. 
    We have that 
    \begin{align}
        \suc^{\rm{MC}}(\cN,M)&\ge  \suc^{\rm{NS}}(\cN,M) 
        \ge \left(1-\frac{1}{M}\right)\cdot \suc^{\rm{MC}}(\cN,M).
    \end{align}
\end{prop}

\begin{proof}
The NS and MC programs are stated in  \eqref{NS-program-pre} and \eqref{MC-program} respectively. 
\sloppy The inequality $\suc^{\rm{MC}}(\cN,M)\ge  \suc^{\rm{NS}}(\cN,M)$ is clear. Let us prove the remaining inequality $      \suc^{\rm{NS}}(\cN,M) \ge \left(1-\frac{1}{M}\right)\cdot \suc^{\rm{MC}}(\cN,M).$ To this end, we consider $(\rho_R,\Lambda_{RB})$ an optimal solution of \eqref{MC-program} with size $M-1$. We construct $ \Lambda'_{RB} = \Lambda_{RB} + \rho_R\otimes (\dI_B-\Lambda_B)$ that satisfies:
\begin{align}
    \Lambda'_{RB}&\mge \Lambda_{RB}\mge 0,
    \\\Lambda'_{B} &= \Lambda_{B} + (\dI_B-\Lambda_B)=\dI_B,
    \\\Lambda'_{RB} &= \Lambda_{RB} + \rho_R\otimes (\dI_B-\Lambda_B)
    \\&\mle (M-1)\rho_R \otimes  \mathbb{I}_B+ \rho_R\otimes \dI_B 
    \\&= M\rho_R \otimes  \mathbb{I}_B.
\end{align}
So $(\rho_R,\Lambda'_{RB})$ is a feasible solution to the program \eqref{NS-program-pre} of size $M$. Hence, we have
\begin{align}
    \suc_{}^{\rm{NS}}(\cN,M)&\ge \frac{1}{M} \tr{\Lambda'_{RB} \cdot (J_{\cN})_{RB}} 
    \\&\ge \frac{1}{M}\tr{\Lambda_{RB} \cdot (J_{\cN})_{RB}}
    \\&= \frac{M-1}{M}\cdot \suc_{}^{\rm{MC}}(\cN,M-1)
    \\&\ge \left(1-\frac{1}{M}\right)\cdot\suc_{}^{\rm{MC}}(\cN,M),
\end{align}
where we used in the last inequality $\suc_{}^{\rm{MC}}(\cN, M-1)\ge \suc_{}^{\rm{MC}}(\cN,M)$ that can be easily checked with the following formulation obtained from \eqref{MC-program} by making the change of variable $\Lambda_{RB}\leftarrow \frac{1}{M} \Lambda_{RB}$:
\begin{align}
    \suc^{\rm{MC}}(\cN, M) = \sup_{\rho_R, \Lambda_{RB}} \quad & \tr{\Lambda_{RB} (J_{\cN})_{RB}} \notag \\
\text{subject to} \quad
& \rho_R \in \cS(R), \notag \\
& \Lambda_B \mle   \frac{1}{M} \mathbb{I}_B, \notag \\
& 0 \mle \Lambda_{RB} \mle  \rho_R \otimes \mathbb{I}_B.
\end{align}
\end{proof}
Proposition \ref{prop:NSvsMC} demonstrates that as the number of messages grows, the optimal success probabilities for MC and NS strategies converge. Since the approximation is multiplicative with respect to success probabilities, this rounding result has a direct implication  in terms of the strong converse exponent. The strong converse exponents for the meta-converse and non-signaling error probabilities are defined as follows:
\begin{align}
E^{\rm{MC}}(\cN, r) &= \lim_{n\rightarrow \infty } - \frac{1}{n}\log \suc^{\rm{MC}}(\cN^{\otimes n}, e^{nr}),
 \\E^{\rm{NS}}(\cN, r) &= \lim_{n\rightarrow \infty } - \frac{1}{n}\log\suc^{\rm{NS}}(\cN^{\otimes n}, e^{nr}).
\end{align}
Using Proposition \ref{prop:NSvsMC} we immediately deduce that the meta-converse and  non-signaling strategies have  the same strong converse exponent:

\begin{cor}\label{cor:SCE-MC-NS-pre}
    Let $\cN$ be a quantum channel. For all $r\ge 0$, we have  that
    \begin{align}
   E^{\rm{MC}}(\cN, r)= E^{\rm{NS}}(\cN, r). 
\end{align}
\end{cor}

The meta-converse \eqref{eq:mc-hypothesis} corresponds to the problem of composite hypothesis testing. We deduce the achievability strong converse exponent of the meta-converse from the one established for composite hypothesis testing by \cite{Hayashi2016Oct}:

\begin{prop}\label{prop:SCE-achiev}
Let  $\cN$ be a quantum channel. For all $r\ge 0$, we have  that
 \begin{align}
     E^{\rm{MC}}(\cN, r) &= \lim_{n\rightarrow \infty } - \frac{1}{n}\log \suc^{\rm{MC}}(\cN^{\otimes n}, e^{nr})  
     \\&\le \sup_{\alpha\ge 0}  \frac{\alpha}{1+\alpha}\left( r- \widetilde{\cI}_{1+\alpha}(\cN) \right). 
    \end{align}
\end{prop}

\begin{proof}
We have that from Proposition \ref{prop-dual-MC}:
\begin{align}
\suc^{\rm{MC}}(\cN, M) \notag
&= \sup_{\rho_R \in \cS(R)} \sup_{0 \mle O \mle \dI} \Big\{\, 
 \tr{\rho_R^{1/2} (J_{\cN})_{RB} \rho_R^{1/2} \cdot O} \;\Big|\; \sup_{\sigma_B \in \cS(B)} \tr{(\rho_R \otimes \sigma_B) O} \le \frac{1}{M} \Big\}.
\end{align}

So, we have for $n\in \mathbb{N}$
\begin{align}
&\suc^{\rm{MC}}(\cN^{\otimes n}, e^{nr}) \notag 
\\&= \sup_{\rho_{R^n} \in \cS(R^n)} \sup_{0 \mle O \mle \dI} \Big\{ 
\tr{\rho_{R^n}^{1/2} (J_{\cN})_{R^n B^n}^{\otimes n} \rho_{R^n}^{1/2} \cdot O} \Big|\;  \sup_{\sigma_{B^n} \in \cS(B^n)} \tr{(\rho_{R^n} \otimes \sigma_{B^n}) O} \le e^{-nr} \Big\} \notag \\
&\ge \sup_{\rho_R \in \cS(R)} \sup_{0 \mle O \mle \dI} \Big\{ 
\tr{\big(\rho_R^{1/2} (J_{\cN})_{RB} \rho_R^{1/2}\big)^{\otimes n} \cdot O} \Big|\; \sup_{\sigma_{B^n} \in \cS(B^n)} \tr{(\rho_R^{\otimes n} \otimes \sigma_{B^n}) O} \le e^{-nr} \Big\}
\end{align}
where we choose the state to be iid, i.e., $\rho_{R^n} = (\rho_R)^{\otimes n}$. From  the strong converse result for the composite hypothesis testing problem of \cite[Theorem 14]{Hayashi2016Oct} we deduce that:
\begin{align}
  & -\frac{1}{n} \log\sup_{0\mle O\mle \dI} \left\{\tr{ \big(\rho_R^{1/2} ( J_{\cN})_{RB}\rho_R^{1/2}\big)^{\otimes n} \cdot O}  \,\middle| \;  \sup_{\sigma_{B^n}\in \cS(B^n)}\tr{(\rho_R^{\otimes n}\otimes \sigma_{B^n}) O}\le e^{-nr} \right\} 
   \\&\underset{n\rightarrow \infty}{\longrightarrow} \sup_{\alpha>0} \frac{\alpha}{1+\alpha}  \left( r - \inf_{\sigma_B \in \cS(B)}\widetilde{I}_{1+\alpha}(\rho_R^{1/2} ( J_{\cN})_{RB}\rho_R^{1/2} \| \rho_R \otimes \sigma_B)\right).
\end{align}
Hence, we get 
\begin{align}
   & \limsup_{n\rightarrow \infty} -\frac{1}{n}\log\suc_{}^{\rm{MC}}(\cN^{\otimes n},e^{nr}) 
   \\&\le \inf_{\rho_{R}\in \cS(R)} \sup_{\alpha\ge 0} \frac{\alpha}{1+\alpha}  \left( r - \inf_{\sigma_B \in \cS(B)}\widetilde{I}_{1+\alpha}(\rho_R^{1/2} ( J_{\cN})_{RB}\rho_R^{1/2} \| \rho_R \otimes \sigma_B)\right)
       \\&= \sup_{\alpha\ge 0} \frac{\alpha}{1+\alpha}  \left( r -\widetilde{I}_{1+\alpha}(\cN)\right),
\end{align}
where the last equality can be proven using Sion's minimax theorem (see \cite[Eqs.~(100)--(102)]{Li2023Nov}).

\end{proof}

On the other hand, the converse strong converse exponent is proven by \cite{Gupta2015Mar} for entanglement-assisted strategies.

\begin{prop}\label{prop:SCE-converse}
     Let  $\cN$ be a quantum channel. For all $r\ge 0$, we have  that
 \begin{align}
     E^{\rm{MC}}(\cN, r) &= \lim_{n\rightarrow \infty } - \frac{1}{n}\log \suc^{\rm{MC}}(\cN^{\otimes n}, e^{nr})  
     \\&\ge \sup_{\alpha\ge 0}  \frac{\alpha}{1+\alpha}\left( r- \widetilde{\cI}_{1+\alpha}(\cN) \right). 
    \end{align}
\end{prop}

This result is stated in \cite{Gupta2015Mar} for EA strategies. By inspection, the proof applies for the meta-converse as well. For the reader's convenience, we include a detailed proof of the converse result of the  meta-converse  strong converse exponent  in Appendix \ref{app:proof-converse}. By combining Corollary \ref{cor:SCE-MC-NS-pre}, Proposition \ref{prop:SCE-achiev}, and Proposition \ref{prop:SCE-converse}, we deduce the  NS and MC strong converse  exponent.

\begin{cor}\label{cor:SCE-MC-NS}
Let $\cN$ be a quantum channel. For all $r\ge 0$, we have  that 
 \begin{align}
     E^{\rm{MC}}(\cN, r) &= E^{\rm{NS}}(\cN, r)
     =\sup_{\alpha\ge 0}  \frac{\alpha}{1+\alpha}\left( r- \widetilde{\cI}_{1+\alpha}(\cN) \right). 
    \end{align}
\end{cor}


\section{Approximation algorithms in the quantum-classical setting}\label{sec:qc}

In this section, we consider the restricted class of quantum-classical channels with classical output system.  We are able to prove approximation results that relate the entanglement-assisted and non-signaling success probabilities.

\begin{prop}\label{prop:rounding-qc}
Let $\cN$ be a quantum-classical channel and $M, M'\ge 1$. We have the following inequality between the entanglement-assisted and non-signaling success probabilities
  \begin{align}
       &\suc^{\rm{EA}}(\cN,M')
       \ge  \frac{M}{M'}\left(1-\left(1-\frac{1}{M}\right)^{M'}\right) \suc^{\rm{NS}}(\cN,M).
    \end{align}
    In particular, when $M'=M$, we obtain 
     \begin{align}
       \suc^{\rm{EA}}(\cN,M)
       &\ge  \left(1-\left(1-\frac{1}{M}\right)^{M}\right) \suc^{\rm{NS}}(\cN,M)
       \\&\ge  \left(1-\frac{1}{e}\right) \suc^{\rm{NS}}(\cN,M).
    \end{align}
    Moreover, when $M'=o(M)$, we have that 
    \begin{align}
        \suc^{\rm{EA}}(\cN,M')
        &\ge \left( 1-\frac{M'}{2M} + o\Big(\frac{M'}{M}\Big)\right) \suc^{\rm{NS}}(\cN,M).
    \end{align}
\end{prop}

This result extends the classical rounding results from \cite{barman2017algorithmic} to the quantum-classical setting with entanglement assistance while maintaining the same approximation guarantees. Additionally, \cite{barman2017algorithmic} proved that similar approximation ratio (between shared randomness and non-signaling success probabilities) is optimal for classical channels.  Since $\suc^{\rm{NS}}(\cN,M) \ge \suc^{\rm{EA}}(\cN,M) \ge \left(1-\frac{1}{e}\right) \suc^{\rm{NS}}(\cN,M)$, we conclude that EA and NS have the same strong converse exponent. Therefore, by Corollary \ref{cor:SCE-MC-NS}, we obtain  the following characterization of the EA strong converse exponent.

\begin{cor}
    Let  $\cN$ be a quantum-classical channel. For all $r\ge 0$, we have  that 
 \begin{align}
     E^{\rm{EA}}(\cN, r) &= E^{\rm{NS}}(\cN, r)
     =\sup_{\alpha\ge 0}  \frac{\alpha}{1+\alpha}\left( r- \widetilde{\cI}_{1+\alpha}(\cN) \right). 
    \end{align}
\end{cor}

\begin{proof}[Proof of Proposition \ref{prop:rounding-qc}]
Let $(\Lambda_{RB}, \rho_{R})$ be an optimal solution to  the NS success probability program \eqref{NS-program}. 
Let $M'$ be the size of the messages to be transmitted with an entanglement-assisted strategy. 
We consider the following scheme inspired from the  position-based coding of \cite{Anshu2018Jun} and the sequential decoder of \cite{Wilde2013Sep}:
\begin{itemize}
    \item[\textbf{Shared entanglement.}] Let $\ket{\phi}_{RR'} = \rho_R^{1/2}\ket{w}_{RR'}$ be a purification of $\rho_{R}$ where $R'\simeq R$. The shared entanglement state is 
    \begin{equation}
        \bigotimes_{m=1}^{M'} \phi_{R_mR_m'},
    \end{equation}
    where $R'_1\cdots R'_{M'}$ systems are held by Alice while $R_1\cdots R_{M'}$ systems are held by Bob.
    \item[\textbf{Encoding.}] To send the message $m\in [M']$, Alice places  the system $R'_m$ in $A'$ and traces out the rest of her part of the shared entanglement.  In other words, she applies the map $\id_{R'_m\rightarrow A'}\otimes \bigotimes_{l\neq m}\ptr{R'_{l}}{\cdot}$. The resulting state is 
    \begin{align}
        \sigma_{R_1R_2\cdots R_{M'} A'}(m) &= \id_{R'_m\rightarrow A'}\otimes \bigotimes_{l\neq m}\ptr{R'_{l}}{ \textstyle \bigotimes_{l=1}^{M'} \phi_{R_{l}R_{l}'} } 
        = \phi_{R_m A'}\otimes \bigotimes_{l\neq m} \rho_{R_l}.   
    \end{align}
    \item[\textbf{Transmission.}] Alice transmits $A'$ over the channel $\cN_{A'\rightarrow B}$. The global state is then 
    \begin{align}
        \zeta_{R_1R_2\cdots R_{M'} B}(m) &= \cN_{A'\rightarrow B}(\sigma_{R_1R_2\cdots R_{M'} A'}(m))
        =\cN_{A'\rightarrow B}\left(\phi_{R_m A'}\right)\otimes \bigotimes_{l\neq m} \rho_{R_l}.
    \end{align}
    \item[\textbf{Decoding.}] Since $(\Lambda_{RB}, \rho_{R})$ is a solution of the program \eqref{NS-program}, they satisfy $0\mle \Lambda_{RB}\mle \rho_R\otimes \dI_B$,  thus we can define the observable
    \begin{align}
        0\mle O_{RB} = \rho_R^{-1/2} \Lambda_{RB}\rho_R^{-1/2} \mle \dI_{RB}.
    \end{align}
    Bob performs the measurement:
    \begin{align}
       & \left\{ \Xi_{R_1R_2\cdots R_{M'} B}(m) \right\}_{m=1}^{M'} \bigcup \left\{ \Xi_{R_1R_2\cdots R_{M'} B}^{\star} \right\}, \quad \text{with }\label{eq:POVM}
        \\ & \Xi_{R_1R_2\cdots R_{M'} B}(m) = \sqrt{\dI-O_{R_{1}B}}\cdots\sqrt{\dI-O_{R_{m-1}B}}O_{R_mB}\sqrt{\dI-O_{R_{m-1}B}}\cdots \sqrt{\dI-O_{R_{1}B}},
        \\ &\quad \;\;  \Xi_{R_1R_2\cdots R_{M'} B}^{\star} = \dI_{R_1R_2\cdots R_{M'} B} - \sum_{m=1}^{M'}\Xi_{R_1R_2\cdots R_{M'} B}(m).\label{eq:POVM-last}
    \end{align}
    This measurement corresponds to sequentially measuring $R_tB$ using the POVM $\{O_{R_tB}, \dI- O_{R_tB}\}$ for $t=1,2, \dots$ and returning the first index $t$ for which we get the first outcome (corresponding to the observable $O_{R_tB}$). 
    Note that the last operator in \eqref{eq:POVM} can be written as
    \begin{align}
    \Xi_{R_1\cdots R_{M'} B}^{\star}=\dI\!-\!\sqrt{\dI\!-\!O_{R_{1}B}}\cdots\sqrt{\dI\!-\!O_{R_{M'-1}B}}(\dI\!-\!O_{R_{M'}B})\sqrt{\dI\!-\!O_{R_{M'-1}B}}\cdots \sqrt{\dI\!-\!O_{R_{1}B}},
    \end{align}
    which shows that it is indeed an observable (using $0\mle O_{RB}\mle \dI_{RB}$) and the set  \eqref{eq:POVM} is a valid POVM. 
   If we observe an outcome corresponding to the operator $\Xi_{R_1R_2\cdots R_{M'} B}^{\star}$ we return $0$.
    Since $B$ is a classical system we have that
    \begin{equation}
       \forall l\neq l' : \quad  O_{R_{l}B}O_{R_{l'}B} = O_{R_{l'}B}O_{R_{l}B},
    \end{equation}
   thus the measurement operators can be written as 
    \begin{align}
        \Xi_{R_1R_2\cdots R_{M'} B}(m) &= \sqrt{\dI-O_{R_{1}B}}\cdots\sqrt{\dI-O_{R_{m-1}B}}O_{R_mB}\sqrt{\dI-O_{R_{m-1}B}}\cdots \sqrt{\dI-O_{R_{1}B}}
        \\&= (\dI-O_{R_{1}B})\cdots(\dI-O_{R_{m-1}B})O_{R_mB}.\label{eq:povm-simple}
    \end{align}
\end{itemize}

\textbf{Analysis of the scheme.} 
The success probability of decoding the $m^{\text{th}}$ message can be calculated as follows:  
\begin{widetext}
    \begin{align}
        &\suc(\cN, M',m) \\&= \tr{\zeta_{R_1R_2\cdots R_{M'} B}(m) \cdot \Xi _{R_1R_2\cdots R_{M'} B}(m)  }
        \\&= \tr{\cN_{A'\rightarrow B}\left(\phi_{R_m A'}\right)\otimes \bigotimes_{l\neq m} \rho_{R_l}\cdot (\dI-O_{R_{1}B})\cdots(\dI-O_{R_{m-1}B})O_{R_mB}}
        \\&= \tr{ \rho_{R_m}^{1/2}\cN_{A'\rightarrow B}\left(w_{R_m A'}\right)\rho_{R_m}^{1/2}\otimes \bigotimes_{l\neq m} \rho_{R_l}\cdot (\dI-O_{R_{1}B})\cdots(\dI-O_{R_{m-1}B})O_{R_mB}}
           \\&= \tr{ (J_{\cN})_{R_m B}\cdot (\rho_{R_1}-\Lambda_{R_{1}B})\cdots(\rho_{R_{m-1}}-\Lambda_{R_{m-1}B})\Lambda_{R_mB}}
        \\&= \tr{ (J_{\cN})_{R_m B}\cdot (\dI_B-\Lambda_{B})\cdots(\dI_B-\Lambda_{B})\Lambda_{R_mB}}
        \\&= \left(1-\frac{1}{M}\right)^{m-1}\tr{ (J_{\cN})_{R_m B}\cdot \Lambda_{R_mB}}
        \\&= \left(1-\frac{1}{M}\right)^{m-1}\suc^{\rm{NS}}(\cN,M),\label{eq:sequential-exact}
    \end{align}
\end{widetext}
    where we used $\Lambda_B = \frac{1}{M}\dI_B$ and $ \tr{ (J_{\cN})_{R B}\cdot \Lambda_{RB}} = \suc^{\rm{NS}}(\cN,M)$. 
\\ Hence, the average success  probability of this scheme can be calculated as follows:
  \begin{align}
      \frac{1}{M'}\sum_{m=1}^{M'} \suc(\cN, M',m)
     &=   \frac{1}{M'}\sum_{m=1}^{M'} \left(1-\frac{1}{M}\right)^{m-1}\suc^{\rm{NS}}(\cN,M)
       \\&=\frac{M}{M'}\left(1-\left(1-\frac{1}{M}\right)^{M'}\right) \suc^{\rm{NS}}(\cN,M). 
    \end{align}
    Finally we deduce the desired bound on the EA success probability:
     \begin{align}
      \suc^{\rm{EA}}(\cN,M')
     &\ge \frac{1}{M'}\sum_{m=1}^{M'}\suc(\cN, M',m)
       \\&= \frac{M}{M'}\left(1-\left(1-\frac{1}{M}\right)^{M'}\right) \suc^{\rm{NS}}(\cN,M).
    \end{align}
\end{proof}

The optimality of Proposition \ref{prop:rounding-qc} remains unclear, as entanglement-assisted  strategies are more challenging to manage compared to shared-randomness strategies (e.g., \cite{Pironio2010May,Berta2016Aug,berta2024optimality}). Moreover, we are unable to establish similar approximations as in Proposition \ref{prop:rounding-qc} in the fully quantum setting. The main challenge is that, in the quantum setting, we can no longer simplify the measurement operators as done in \eqref{eq:povm-simple}. Consequently, the analysis of the sequential decoder \eqref{eq:POVM} becomes non-trivial and does not yield the exact success probabilities as in \eqref{eq:sequential-exact}. In the next section, we present  alternative approximation results for channel coding in the quantum setting.


\section{Approximation algorithms in the quantum setting}
\label{sec:qq}

In this section, we explore approximation algorithms for channel coding in the general quantum setting. Our objective is to design algorithms that approximate  the entanglement-assisted  success probability. We begin by presenting  a simple approximation with an additive error in Section \ref{sec:qq-additive}, followed by a more involved approximation that achieves a multiplicative error in Section \ref{sec:qq-multiplicative}. 

\subsection{Approximation with an additive error}\label{sec:qq-additive}


Our first attempt at designing approximation algorithms uses the same shared entanglement and the position-based coding \cite{Anshu2018Jun,qi18}, similar to what we have done in the quantum-classical setting in Section \ref{sec:qc}. However, to deal with the non-commutative difficulty that arises in the quantum setting, we use the well-known decoder and inequality of Hayashi-Nagaoka \cite{Hayashi2003Jun}.

\begin{prop}\label{prop:rounding-qq-additive}
Let $\cN$ be quantum channel and $M, M'\ge 1$. We have the following inequality between the entanglement-assisted and non-signaling error probabilities
 \begin{align}
    \eps_{}^{\rm{EA}}(\cN, M')\le 2\eps_{}^{\rm{NS}}(\cN,M) +4\frac{M'-1}{M}.
\end{align}
Moreover, we have the following inequality between the entanglement-assisted and non-signaling success probabilities 
\begin{align}
   \suc_{}^{\rm{EA}}(\cN, M')&\ge   \suc^{\rm{NS}}(\cN,M)  -5 \sqrt{\frac{M'}{M}}. 
\end{align}
\end{prop}

This proposition is similar to what was achieved for classical-quantum channels using shared-randomness strategies by \cite{Fawzi19}. Additionally, it implies that  non-signaling correlations do not allow to improve the entanglement-assisted classical capacity of a quantum channel \cite{Leung2015Jun}. Finally, for an approximation in terms of error probabilities, the additive error $O(\frac{M'}{M})$ seems to be unavoidable. This can be seen by the optimality proof  in the classical setting of \cite{barman2017algorithmic}. In terms of success probabilities, the additive error $O\big(\sqrt{\frac{M'}{M}}\big)$ is not sufficient enough to deduce optimal strong converse exponents. For this reason, in the next section,  we explore more involved approximation algorithms that achieve small multiplicative errors. 

\begin{proof}[Proof of Proposition \ref{prop:rounding-qq-additive}]
Let $(\Lambda_{RB}, \rho_{R})$ be an optimal solution of the NS  success probability program \eqref{NS-program}.  Let $M'$ be the number of the messages  to be transmitted with an entanglement-assisted strategy. 
We consider the following scheme inspired by the position-based coding of \cite{Anshu2018Jun} and Hayashi-Nagaoka decoder \cite{Hayashi2003Jun}:

\begin{itemize}
    \item[\textbf{Shared entanglement.}] Let $\ket{\phi}_{RR'} = \rho_R^{1/2}\ket{w}_{RR'}$ be a purification of $\rho_{R}$ where $R'\simeq R$. The shared entanglement state is 
    \begin{equation}
        \bigotimes_{m=1}^{M'} \phi_{R_mR_m'},
    \end{equation}
    where $R'_1\cdots R'_{M'}$ systems are held by Alice while $R_1\cdots R_{M'}$ systems are held by Bob.
    \item[\textbf{Encoding.}] To send the message $m\in [M']$, Alice places the system $R'_m$ in $A'$ and traces out the rest of her part of the shared entanglement.  In other words, she applies the map $\id_{R'_m\rightarrow A'}\otimes \bigotimes_{l\neq m}\ptr{R'_{l}}{\cdot}$. The resulting state is 
    \begin{align}
        \sigma_{R_1R_2\cdots R_{M'} A'}(m) &= \id_{R'_m\rightarrow A'}\otimes \bigotimes_{l\neq m}\operatorname{Tr}_{R'_{l}}\bigg[\bigotimes_{l=1}^{M'} \phi_{R_{l}R_{l}'} \bigg] 
         =\phi_{R_m A'}\otimes \bigotimes_{l\neq m} \rho_{R_l}.   
    \end{align}
    \item[\textbf{Transmission.}] Alice transmits $A'$ over the channel $\cN_{A'\rightarrow B}$. The global state is then 
    \begin{align}
        \zeta_{R_1R_2\cdots R_{M'} B}(m) &= \cN_{A'\rightarrow B}(\sigma_{R_1R_2\cdots R_{M'} A'}(m))
        =\cN_{A'\rightarrow B}\left(\phi_{R_m A'}\right)\otimes \bigotimes_{l\neq m} \rho_{R_l}.
    \end{align}
    \item[\textbf{Decoding.}] Since $(\Lambda_{RB}, \rho_{R})$ is a solution of the program \eqref{NS-program}, they satisfy $0\mle \Lambda_{RB}\mle \rho_R\otimes \dI_B $ thus we can define the observable
    \begin{align}
        0\mle O_{RB} = \rho_R^{-1/2} \Lambda_{RB}\rho_R^{-1/2} \mle \dI_{RB}.
    \end{align}
    Bob performs the measurement \cite{Hayashi2003Jun}
    \begin{align}
        \left\{ \Xi_{R_1R_2\cdots R_{M'} B}(m) = \left(\textstyle\sum_{l=1}^{M'} O_{R_lB}\right)^{-1/2}O_{R_mB} \left(\textstyle\sum_{l=1}^{M'} O_{R_lB}\right)^{-1/2}\right\}_{m=1}^{M'}.
    \end{align}
\end{itemize}

\textbf{Analysis of the scheme.} We rely on the following inequality \cite{Hayashi2003Jun} valid for all $0\mle A \mle \dI$, $B\mge 0$ and $c>0$:
\begin{align}
    &\dI - (A+B)^{-1/2}A(A+B)^{-1/2}
    \mle (1+c)(\dI-A) + (2+c+c^{-1})B.
\end{align}
Using this inequality, the error probability of decoding the $m^{\text{th}}$ message can be controlled  as follows:   
    \begin{align}
        \eps(\cN, M', m) 
        &= \tr{\zeta_{R_1\cdots R_{M'} B}(m) \cdot (\dI -\Xi (m))_{R_1R_2\cdots R_{M'} B} }
        \\&\le (1+c)\tr{\zeta_{R_1\cdots R_{M'} B}(m) \cdot (\dI-O_{R_mB})  } \label{eq:error-HN1}
        \\&\;\; +\!(2+c+c^{-1}) \tr{\zeta_{R_1\cdots R_{M'} B}(m)  ( \textstyle\sum_{l\neq m} O_{R_lB} )  }
        \label{eq:error-HN2}
    \end{align}
Recall that $\ket{\phi}_{R A'}=\rho_{R}^{1/2}\ket{w}_{RA'} $.
The first term  \eqref{eq:error-HN1} can be simplified as follows: 
\begin{align}
    \tr{\zeta_{R_1R_2\cdots R_{M'} B} (m)\cdot (\dI-O_{R_mB}) } 
    &= 1-\tr{\cN_{A'\rightarrow B}\left(\phi_{R_m A'}\right)\otimes \bigotimes_{l\neq m} \rho_{R_l}\cdot  \frac{ \Lambda_{R_mB}}{\rho_{R_m}}}
    \\&= 1-\tr{ \rho_{R_m}^{1/2}\cN_{A'\rightarrow B}\left(w_{R_m A'}\right) \rho_{R_m}^{1/2}\cdot  \frac{ \Lambda_{R_mB}}{\rho_{R_m}}}
    \\&=1- \tr{(J_{\cN})_{R_mB} \cdot \Lambda_{R_mB}}
    \\&= \eps_{}^{\rm{NS}}(\cN,M). \label{eq:HN-first-term}
\end{align}
The second term  \eqref{eq:error-HN2} can be calculated as follows: 
\begin{align}
    \tr{\zeta_{R_1R_2\cdots R_{M'} B}(m) \cdot (\textstyle\sum_{l\neq m} O_{R_lB} )}
   &= \tr{\cN_{A'\rightarrow B}\left(\phi_{R_m A'}\right)\otimes \bigotimes_{l\neq m} \rho_{R_l} \cdot (\textstyle\sum_{l\neq m} \frac{ \Lambda_{R_lB}}{\rho_{R_l}} ) }
    \\&= \textstyle\sum_{l\neq m} \tr{ \cN_{A'\rightarrow B}\left(\rho_{ A'}\right) \cdot  \Lambda_{B}} 
    \\&=  \textstyle\sum_{l\neq m} \tr{ \cN_{A'\rightarrow B}\left(\rho_{ A'}\right) \cdot  \frac{1}{M}\dI_B}
    \\&= \frac{M'-1}{M}. \label{eq:HN-second-term}
\end{align}
Hence from \eqref{eq:error-HN1}, \eqref{eq:error-HN2}, \eqref{eq:HN-first-term} and \eqref{eq:HN-second-term} we deduce that the average error probability is bounded as follows:
\begin{align}\label{eq:err-c}
   & \eps_{}^{\rm{EA}}(\cN, M')\le \frac{1}{M'}\sum_{m=1}^{M'}\eps(\cN, M', m)
    \le (1+c)\eps_{}^{\rm{NS}}(\cN,M) +(2+c+c^{-1}) \frac{M'-1}{M}.
\end{align}
In particular for $c=1$ we deduce 
\begin{align}
    \eps_{}^{\rm{EA}}(\cN, M')\le 2\eps_{}^{\rm{NS}}(\cN,M) +4\frac{M'-1}{M}.
\end{align}
In terms of success probability we have that from \eqref{eq:err-c}:
\begin{align}
   1- \eps_{}^{\rm{EA}}(\cN, M')&\ge  1- \eps_{}^{\rm{NS}}(\cN,M) -c   \eps_{}^{\rm{NS}}(\cN,M)   
   -(2+c+c^{-1}) \frac{M'-1}{M}.
\end{align}
Then, if $M'\le M$, we choose $c= \sqrt{\frac{M'}{M}}\le 1$ and obtain 
\begin{align}
   \suc^{\rm{EA}}(\cN, M')&\ge \suc^{\rm{NS}}(\cN,M) -2\sqrt{\frac{M'}{M}} -3 \frac{M'}{M}
   \\&\ge  \suc^{\rm{NS}}(\cN,M)  -5 \sqrt{\frac{M'}{M}}. 
\end{align}
If $M'>M$, this inequality holds trivially.
\end{proof}


\subsection{Approximation  with a multiplicative error}\label{sec:qq-multiplicative}

In this section, we focus  on obtaining an approximation of the entanglement-assisted success probability with a small multiplicative error in the quantum setting, as similar results in  the classical, classical-quantum,  and quantum-classical settings have already been established in \cite{barman2017algorithmic,Fawzi19}, and in Section \ref{sec:qc},  respectively. To this end, we round a non-signaling strategy to an entanglement-assisted one. The program of the non-signaling success probability \eqref{NS-program} can be written as 
\begin{align}
           \suc^{\rm{NS}}(\cN,M)
           &=\sup_{\rho_{R} \in \cS(R)} \sup_{\Lambda_{RB}} \left\{\tr{\Lambda_{RB} \cdot (J_{\cN})_{RB}}\,\middle|\,\Lambda_B = \tfrac{1}{M}\mathbb{I}_B,\, 0\mle  \Lambda_{RB} \mle \rho_R \otimes  \mathbb{I}_B\right\} \\
           &= \sup_{\rho_{R} \in \cS(R)} \suc^{\rm{NS}}(\cN,M, \rho),
\end{align}
where we introduce $\suc^{\rm{NS}}(\cN,M, \rho) $ for $\rho_{R} \in \cS(R)$ defined as:
\begin{align}\label{eq:ns-program-rho}
     \suc^{\rm{NS}}(\cN,M,  \rho) = \sup_{ \Lambda_{RB}} \quad & \tr{\Lambda_{RB} (J_{\cN})_{RB}} \notag \\
\text{subject to} \quad
& \Lambda_B = \frac{1}{M} \mathbb{I}_B, \notag \\
& 0 \mle \Lambda_{RB} \mle \rho_R \otimes \mathbb{I}_B. 
\end{align}
We have the following rounding result of the non-signaling success probability in the quantum setting.

\begin{prop}\label{prop:qq-rounding}
Let $\cN_{A\rightarrow B}$ be a quantum channel,   $\rho_R \in \cS(R)$ be a quantum state and $M\ge 1$. 
Let $v$ be the number of distinct eigenvalues of $\rho_R$. Assume that $\log|A|\ge e^2$. 
    We have that
    \begin{align}\label{eq:prefactor}
         \suc^{\rm{EA}}(\cN,M) &\ge \frac{1}{2v \log(2v M e^{v}|A|^{v}|B|)}\cdot \suc^{\rm{NS}}(\cN, M,  \rho) .
    \end{align}
\end{prop}

This approximation of the success probability is not as clean as the one established for quantum-classical channels in Proposition~\ref{prop:rounding-qc}. However, since the multiplicative error in Eq.~\eqref{eq:prefactor} depends logarithmically on the size of the channel $\cN$, the communication size, and the number of distinct eigenvalues of the state $\rho_R$, we will see in Section \ref{sec:qq-application-sce} that in terms of the strong converse exponent, this approximation is largely sufficient.\\

To prove Proposition \ref{prop:qq-rounding}, we combine ideas from the classical-quantum rounding result of \cite{Fawzi19}, coding using the decoupling technique \cite{Hsieh2008Jun,Li2023Nov}, position-based coding \cite{Anshu2018Jun}, and the flattening input state technique \cite{Li2023Nov}. 
\begin{proof}[Proof of Proposition \ref{prop:qq-rounding}]Let $A' \simeq R\simeq A$. Write $R= \bigoplus_{j=1}^{v } R_j$ and
    \begin{align}
        \rho_R = \bigoplus_{j=1}^{v} \lambda_j \Pi^j_R = \bigoplus_{j=1}^{v} p_j \mathbf{1}_{R_j},
    \end{align}
where $\{\lambda_j\}_j$ are the distinct eigenvalues of $\rho_R$, $\{\Pi^j_R\}_j$ are orthogonal projectors onto the Hilbert spaces $\{R_j\}_j$, $\mathbf{1}_{R_j}$ is the maximally mixed state on $R_j$, i.e., $\mathbf{1}_{R_j}=\frac{1}{|R_j|}\Pi^j_R$ and $p_j = |R_j|\cdot \lambda_j$.

Let $\cL(\cH)$ be any unitary $1$-design probability measure on $\mathbb{U}(\cH)$, i.e, for all $X_{\cH \cK}$
\begin{align}
    \exs{U \sim \cL(\cH)}{U_{\cH}X_{\cH \cK}U_{\cH}^\dagger}=\mathbf{1}_{\cH} \otimes X_\cK. 
\end{align}
Sampling uniformly at random from the  set of Pauli operators is an example of a unitary $1$-design. 

Denote by $\Psi_{A A'} = \proj{\Psi}_{A A'}$ the maximally entangled state on $AA'$. Concretely, if $\{\ket{x_i}\}_{i=1}^{|A|}$ is an orthonormal basis of $A$ consisting of the eigenvectors of $\rho$, $\ket{\Psi}_{A A'}$ is defined as:
\begin{align}
    \ket{\Psi}_{A A'}= \frac{1}{\sqrt{|A|}}\sum_{i=1}^{|A|} \ket{x_i}_{A}\otimes \ket{x_i}_{A'}.
\end{align}
Consider the following rounding scheme illustrated in Figure \ref{fig:qq}.

\begin{figure}[t]
    \centering
    \includegraphics[width=0.7\linewidth]{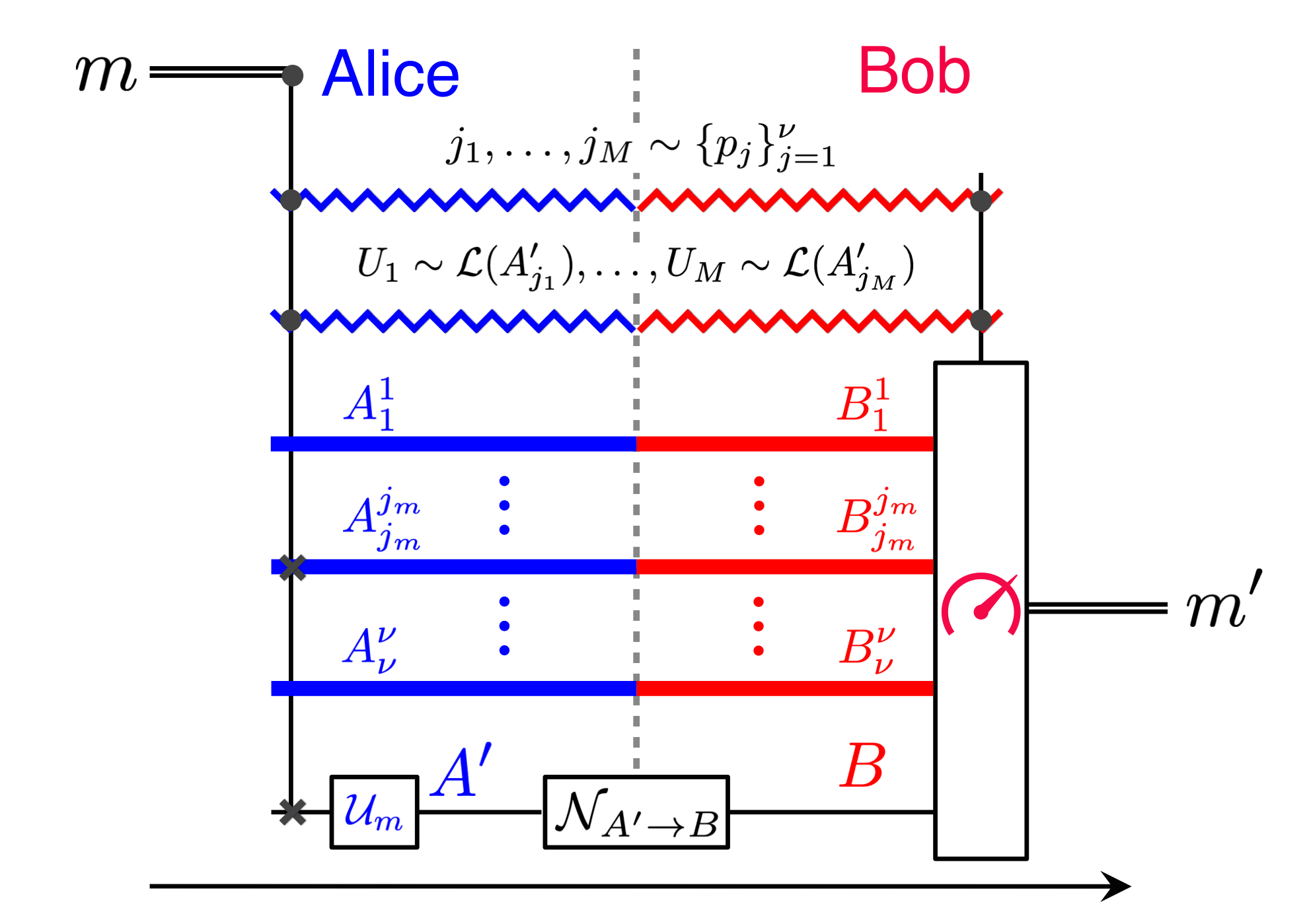}
    \caption{Illustration of an entanglement-assisted protocol achieving the performance of Proposition \ref{prop:qq-rounding}. \textit{Wavy lines:} Represent the shared randomness. \textit{Thick lines:} Represent the shared entanglement.  \textit{Arrow:} Indicates the time progression of the protocol. Depending on the message $m$ and the shared randomness, the system $A_{j_m}^{j_m}$ is swapped with $A'$. Then, the unitary channel $\cU_m (\cdot)= U_m^\top(\cdot)\overline{U}_m$ (which depends on $m$ and the shared randomness) is applied before the transmission of the system $A'$ through the channel $\cN_{A'\rightarrow B}$. In order to decode the message, a measurement (which depends on the shared randomness) is performed on the systems $B^1_1\cdots B_{v}^v B$ yielding the outcome $m'$.}
    \label{fig:qq}
\end{figure}
\begin{itemize}
    \item[\textbf{Shared entanglement and classical randomness.}] Let  $j_1, \dots, j_M \overset{\rm{iid}}{\sim} \{p_j\}_{j}$, $U_1 \sim \cL(R_{j_1}), \dots, U_M \sim \cL(R_{j_M})$ and we write $\bm{j}\sim p, \bm{U}\sim \cL$. This forms the shared-randomness, and  the shared entanglement is given by
    \begin{align}
        \bigotimes_{j=1}^{v} \Psi_{B^j_{j}A^j_{j}},
    \end{align}
    where $A^1\simeq \cdots \simeq A^{v} \simeq R$ systems are held by Alice while $B^1\simeq \cdots \simeq B^{v} \simeq R$ systems are held by Bob. For $j\in [v]$, $A^j_j$ is the $j^{\text{th}}$ subspace in the decomposition of $A^j = \bigoplus_{l=1}^{v} A_l^j$, and similarly, $B^j_j$ is the $j^{\text{th}}$ subspace in the decomposition of $B^j = \bigoplus_{l=1}^{v} B_l^j$.
    \item[\textbf{Encoding.}] 
    To send the message $m\in [M]$, Alice first applies  the unitary $\cU_m(\cdot) = U_m^{\top} (\cdot) \overline{U}_m$ on $A_{j_m}^{j_m}$,  then she places $A_{j_m}^{j_m}$ in $A'$.  The resulting state in the $j_m^{\text{th}}$ tensor factor of the shared entanglement tensor product is 
    \begin{align}
        \Psi_{B_{j_m}^{j_m}A_{j_m}^{j_m}} &\;\;\,\overset{\cU_m}{\longrightarrow}\,\;\; (U_m^\top)_{A_{j_m}^{j_m}}  \Psi_{B_{j_m}^{j_m}A_{j_m}^{j_m}}(\overline{U}_t)_{A_{j_m}^{j_m}} 
        \\& \;\;\; \; = \; \;\;\;(U_m)_{B_{j_m}^{j_m}}  \Psi_{B_{j_m}^{j_m}A_{j_m}^{j_m}}(U_m^\dagger)_{B_{j_m}^{j_m}}
        \\&\overset{A_{j_m}^{j_m}\rightarrow A'}{\longrightarrow} \iota_{A^{j_m}\rightarrow  A'}\left((U_m)_{B_{j_m}^{j_m}}  \Psi_{B_{j_m}^{j_m}A_{j_m}^{j_m}}(U_m^\dagger)_{B_{j_m}^{j_m}}\right),
    \end{align}
    where we used the property of the maximally entangled state $U_{B^j}\otimes \overline{U}_{A^j}\ket{\Psi}_{B^jA^j} = \ket{\Psi}_{B^jA^j}$, and  $\iota_{A^{j_m} \rightarrow A'}$ is a unitary channel since $A^{j_m}\simeq R\simeq A'$. Note that the state $\Psi_{B^{j_m}_{j_m}A^{j_m}_{j_m}}$ on $B^{j_m}_{j_m}A^{j_m}_{j_m}$ can be considered as a matrix on $B^{j_m}_{j_m}A^{j_m}$ by inclusion (i.e. by padding zero block matrices).
    Next, Alice traces out her remaining part of the shared entanglement and the global state is then 
    \begin{align}
        \sigma_{B^1 \cdots B^{v} A'}(m) &=  (U_m)_{B_{j_m}^{j_m}} \cdot\iota_{A^{j_m}\rightarrow  A'}\left(\Psi_{B_{j_m}^{j_m}A_{j_m}^{j_m}}\right)\cdot(U_m^\dagger)_{B_{j_m}^{j_m}} \otimes \bigotimes_{s\in [v]\setminus j_m} \Psi_{B^s_{s}} 
        \\&=(U_m)_{B_{j_m}^{j_m}} \cdot\iota_{A^{j_m}\rightarrow  A'}\left(\Psi_{B_{j_m}^{j_m}A_{j_m}^{j_m}}\right)\cdot(U_m^\dagger)_{B_{j_m}^{j_m}} \otimes \bigotimes_{s\in [v]\setminus j_m} \bid_{B^s_{s}}.
    \end{align}
    More precisely, by denoting $\cV_m(\cdot)= U_m(\cdot)U_m^\dagger$, the encoding channel is 
    \begin{align}    \cE^m_{A_{1}^{1}\cdots A_{v}^{v}\rightarrow A'}\left(\cdot \right) &= \iota_{A^{j_m}\rightarrow  A'}\circ (\cU_m)_{A_{j_m}^{j_m}\rightarrow A_{j_m}^{j_m}}  (\cdot) \otimes \bigotimes_{s\in [v]\setminus j_m} \ptr{(B^s_{s})'}{\cdot}
    \\&= (\cV_m)_{B_{j_m}^{j_m}\rightarrow B_{j_m}^{j_m}}\circ \iota_{A^{j_m}\rightarrow  A'}   (\cdot) \otimes \bigotimes_{s\in [v]\setminus j_m} \ptr{(B^s_{s})'}{\cdot},
    \end{align}
   and the global state can be written as 
   \begin{align}
       &\cE^m_{A_{1}^{1}\cdots A_{v}^{v}\rightarrow A'}\left(\bigotimes_{j=1}^{v} \Psi_{B^j_{j} A^j_{j}}\right) 
       \\&= \iota_{A^{j_m}\rightarrow  A'}\circ (\cV_m)_{B_{j_m}^{j_m}\rightarrow B_{j_m}^{j_m}}  (\Psi_{B^{j_m}_{j_m}A^{j_m}_{j_m}}) \otimes \bigotimes_{s\in [v]\setminus j_m} \ptr{(B^s_{s})'}{\Psi_{B^s_{s}(B^s_{s})'} }.
   \end{align}
    \item[\textbf{Transmission.}] Alice transmits the system $A'$ over   the channel $\cN_{A'\rightarrow B}$. The global state becomes 
    \begin{align}
         \zeta_{B^1 \cdots B^{v} B}(m)
        &=(\cV_m)_{B_{j_m}^{j_m}\rightarrow B_{j_m}^{j_m}} \circ \cN_{A'\rightarrow B}\circ \iota_{A^{j_m}\rightarrow  A'}\left(\Psi_{B_{j_m}^{j_m}A_{j_m}^{j_m}}\right) \otimes \bigotimes_{s\in [v]\setminus j_m} \bid_{B^s_{s}}. 
    \end{align}
    \item[\textbf{Decoding.}] Let $\Lambda_{RB}$ be an optimal solution to 
    \begin{align}\label{eq:ns-prog-q-proof}
     \suc^{\rm{NS}}(\cN, M,  \rho_R) &=\sup_{\Lambda_{RB}} \left\{\tr{\Lambda_{RB} \cdot (J_{\cN})_{RB}} | \Lambda_B = \tfrac{1}{M}\mathbb{I}_B, \,0\mle  \Lambda_{RB} \mle \rho_R \otimes  \mathbb{I}_B   \right\}.
\end{align}
    Bob performs the measurement 
    \begin{align}
        &\bigg\{\left\|\textstyle\sum_{l=1}^M O^{\bm{j}, \, \bm{U}}(l)\right\|_{\infty}^{-1}\cdot \,  O^{\bm{j}, \, \bm{U}}(m)\bigg\}_{m=1}^M \;\bigcup\; \bigg\{\dI - \textstyle\sum_{m=1}^{M}\left\|\textstyle\sum_{l=1}^M O^{\bm{j}, \, \bm{U}}(l)\right\|_{\infty}^{-1}\cdot \,  O^{\bm{j}, \, \bm{U}}(m)\bigg\}\,, \label{eq:povm-q}
        \\ & \qquad \text{with }\,   O^{\bm{j}, \, \bm{U}}_{B^1\cdots B^{v}B}(m) = \frac{|B_{j_m}^{j_m}|}{p_{j_m}}\cdot (U_m)_{B_{j_m}^{j_m}}\Pi_{B^{j_m}}^{j_m} \Lambda_{B^{j_m} B} \Pi_{B^{j_m}}^{j_m}(U_m^\dagger)_{B_{j_m}^{j_m}} \otimes \bigotimes_{s\in [v]\setminus j_m} \dI_{B^s}.
    \end{align}
\end{itemize}

Recall that, for $s\in [v]$,  $B^s = \bigoplus_{j=1}^v B^s_j$,   and  $\Pi_{B^{j_m}}^{j_m}$ is the orthogonal projector onto the Hilbert subspace $B_{j_m}^{j_m}$ of $B^{j_m}$.  $O^{\bm{j}, \, \bm{U}}(m)$ should be viewed (by padding zero blocks when necessary) as an operator on the space $B^1B^2\cdots B^{v}B$ so  the dimension of $O^{\bm{j}, \, \bm{U}}(m)$  is  $|R|^{v} \cdot |B|$.  Note that $\textstyle\sum_{m=1}^{M}\left\|\textstyle\sum_{l=1}^M O^{\bm{j}, \, \bm{U}}(l)\right\|_{\infty}^{-1}\cdot \,  O^{\bm{j}, \, \bm{U}}(m)\mle \dI$ which confirms that \eqref{eq:povm-q} is a valid POVM. If Bob observes an outcome  `$m'$', corresponding to the measurement operator $\left\|\textstyle\sum_{l=1}^M O^{\bm{j}, \, \bm{U}}(l)\right\|_{\infty}^{-1}\cdot \,  O^{\bm{j}, \, \bm{U}}(m')$, he decodes the message and returns $m'$. Otherwise, if the outcome corresponds to  $\dI - \textstyle\sum_{m=1}^{M}\left\|\textstyle\sum_{l=1}^M O^{\bm{j}, \, \bm{U}}(l)\right\|_{\infty}^{-1}\cdot \,  O^{\bm{j}, \, \bm{U}}(m)$, Bob returns $0$. 

Let us analyze this scheme. Denote by $Z$ the random variable $Z= {\left\|\sum_{m=1}^M O^{\bm{j}, \, \bm{U}}(m)\right\|_{\infty} }$. The success probability of sending the message $m\in [M]$ conditioned on $(j_1, \dots, j_M)$ and $(U_1, \dots, U_M)$ is:
\begin{align}
   & \suc(\cN, M, m)
   \\&= \tr{\tfrac{1}{Z}\cdot O^{\bm{j}, \, \bm{U}}_{B^1 \cdots B^{v}B}(m) \cdot   \zeta_{B^1 \cdots B^{v} B}(m)}
    \\&= \tfrac{|B_{j_m}^{j_m}|}{Zp_{j_m}}\tr{ (U_m)_{B_{j_m}^{j_m}}\Pi_{B^{j_m}}^{j_m} \Lambda_{B^{j_m} B} \Pi_{B^{j_m}}^{j_m}(U_m^\dagger)_{B_{j_m}^{j_m}}\!\cdot\!   (U_m)_{B_{j_m}^{j_m}} \cN_{A'\rightarrow B}\!\circ\! \iota_{A^{j_m}\rightarrow  A'}\!\left(\!\Psi_{B_{j_m}^{j_m}A_{j_m}^{j_m}}\!\right)\!(U_m^\dagger)_{B_{j_m}^{j_m}} }
    \\&= \tfrac{1}{Zp_{j_m}}\tr{ \Pi_{B^{j_m}}^{j_m} \Lambda_{B^{j_m} B} \Pi_{B^{j_m}}^{j_m} \cdot\cN_{A'\rightarrow B}\circ \iota_{A^{j_m}\rightarrow  A'}\left(|B_{j_m}^{j_m}|\cdot \Psi_{B_{j_m}^{j_m}A_{j_m}^{j_m}}\right) }
       \\&\overset{(a)}{=} \tfrac{1}{Zp_{j_m}}\tr{ \Pi_{B^{j_m}}^{j_m} \Lambda_{B^{j_m} B} \Pi_{B^{j_m}}^{j_m} \cdot\cN_{A'\rightarrow B}\circ \iota_{A^{j_m}\rightarrow  A'}\left(|B^{j_m}|\cdot \Pi_{B^{j_m}}^{j_m}  \Psi_{B^{j_m}A^{j_m}}\Pi_{B^{j_m}}^{j_m} \right) }
      \\&\overset{(b)}{=} \tfrac{1}{Zp_{j_m}}\tr{ \Pi_{B^{j_m}}^{j_m} \Lambda_{B^{j_m} B} \Pi_{B^{j_m}}^{j_m} \cdot\cN_{A'\rightarrow B}\circ \iota_{A^{j_m}\rightarrow  A'}\left(|B^{j_m}| \cdot \Psi_{B^{j_m}A^{j_m}}\right) }
      \\& =\tfrac{1}{Zp_{j_m}}\tr{ \Pi_{B^{j_m}}^{j_m} \Lambda_{B^{j_m} B} \Pi_{B^{j_m}}^{j_m} \cdot(J_{\cN})_{B^{j_m} B} }, \label{eq:suc1}
\end{align}
\sloppy where in $(a)$ we used   $ |B_{j_m}^{j_m}|\cdot \Psi_{B_{j_m}^{j_m}A_{j_m}^{j_m}} = |B^{j_m}|\cdot \Pi_{B^{j_m}}^{j_m} \Psi_{B^{j_m}A^{j_m}}\Pi_{B^{j_m}}^{j_m}$  proved below;  in $(b)$ we used  that $\Pi_{B^{j_m}}^{j_m}$ is a projector onto  $B_{j_m}^{j_m}$ that commutes with $\cN_{A'\rightarrow B}\circ \iota_{A^{j_m}\rightarrow  A'}$. 
\\To prove $ |B_{j_m}^{j_m}|\Psi_{B_{j_m}^{j_m}A_{j_m}^{j_m}} = \Pi_{B^{j_m}}^{j_m}|B^{j_m}| \Psi_{B^{j_m}A^{j_m}}\Pi_{B^{j_m}}^{j_m}$, we note that for a basis element $\ket{x}$ of $B^{j_m}$ (an eigenvector of $\rho$), we have either $\Pi_{B^{j_m}}^{j_m} \ket{x} = 0$ or $\Pi_{B^{j_m}}^{j_m} \ket{x} = \ket{x}$ in which case $\ket{x}\in B_{j_m}^{j_m}$ so 
\begin{align}
   |B^{j_m}|\cdot  \Pi_{B^{j_m}}^{j_m}  \Psi_{B^{j_m}A^{j_m}}\Pi_{B^{j_m}}^{j_m}  &= \Pi_{B^{j_m}}^{j_m}\sum_{i,j=1}^{|R|} \ket{x_i}\bra{x_j}_{B^{j_m}} \otimes \ket{x_i}\bra{x_j}_{A^{j_m}}\Pi_{B^{j_m}}^{j_m}
    \\&= \sum_{\substack{i,j \in [|R|]: \\ x_i,x_j \in B_{j_m}^{j_m}}} \ket{x_i}\bra{x_j}_{B_{j_m}^{j_m}} \otimes \ket{x_i}\bra{x_j}_{A_{j_m}^{j_m}}
    \\&= |B_{j_m}^{j_m}|\cdot \Psi_{B_{j_m}^{j_m}A_{j_m}^{j_m}}.
\end{align}

Let $\cG=\{Z\le \log(2v M e^{v}|R|^{v}|B|)\}$ denotes the ‘good' event that $Z$ is bounded. We have that from \eqref{eq:suc1}:
\begin{align}
    &\exs{\bm{j}\sim p, \bm{U}\sim \cL}{\suc(\cN, M, m)}
    \\&\ge  \exs{\bm{j}\sim p, \bm{U}\sim \cL}{\suc(\cN, M, m)\cdot \bid(\cG)}
    \\&= \exs{\bm{j}\sim p, \bm{U}\sim \cL}{\tfrac{1}{Zp_{j_m}}\tr{ \Pi_{B^{j_m}}^{j_m} \Lambda_{B^{j_m} B} \Pi_{B^{j_m}}^{j_m} \cdot (J_{\cN})_{B^{j_m} B} }\cdot \bid(\cG)}
    \\&\ge  \exs{\bm{j}\sim p, \bm{U}\sim \cL}{\frac{1}{p_{j_m}}\cdot\frac{1}{\log(2v M e^{v}|R|^{v}|B|)}\cdot \tr{ \Pi_{B^{j_m}}^{j_m} \Lambda_{B^{j_m} B} \Pi_{B^{j_m}}^{j_m} \cdot(J_{\cN})_{B^{j_m} B} }\cdot \bid(\cG)}
     \\&=  \exs{\bm{j}\sim p, \bm{U}\sim \cL}{\frac{1}{p_{j_m}}\cdot\frac{1}{\log(2v M e^{v}|R|^{v}|B|)}\cdot \tr{ \Pi_{B^{j_m}}^{j_m} \Lambda_{B^{j_m} B} \Pi_{B^{j_m}}^{j_m} \cdot(J_{\cN})_{B^{j_m} B} }}\label{eq:rdm1}
      \\&\quad -  \exs{\bm{j}\sim p, \bm{U}\sim \cL}{\frac{1}{p_{j_m}}\cdot\frac{1}{\log(2v M e^{v}|R|^{v}|B|)}\cdot \tr{ \Pi_{B^{j_m}}^{j_m} \Lambda_{B^{j_m} B} \Pi_{B^{j_m}}^{j_m} \cdot(J_{\cN})_{B^{j_m} B} }\cdot \bid(\cG^c)}.\label{eq:rdm2}
\end{align}
The expectation \eqref{eq:rdm1} can be lower bounded as follows:
\begin{align}
    &\exs{\bm{j}\sim p, \bm{U}\sim \cL}{\frac{1}{p_{j_m}}\cdot\frac{1}{\log(2v M e^{v}|R|^{v}|B|)}\cdot \tr{ \Pi_{B^{j_m}}^{j_m} \Lambda_{B^{j_m} B} \Pi_{B^{j_m}}^{j_m} \cdot(J_{\cN})_{B^{j_m} B} }}
     \\& =\sum_{j=1}^{v} p_j \frac{1}{p_{j}}\cdot\frac{1}{\log(2v M e^{v}|R|^{v}|B|)} \cdot \tr{ \Pi^{j}_R \Lambda_{R B} \Pi^{j}_R \cdot(J_{\cN})_{RB} }
    \\&= \frac{1}{\log(2v M e^{v}|R|^{v}|B|)}\cdot \operatorname{Tr}\bigg[ \sum_{j=1}^{v} \Pi^{j}_R \Lambda_{R B} \Pi^{j}_R \cdot(J_{\cN})_{RB} \bigg]
    \\&\ge \frac{1}{\log(2v M e^{v}|R|^{v}|B|)}\cdot \frac{1}{v}\cdot \tr{ \Lambda_{R B}  \cdot (J_{\cN})_{RB} }
    \\&= \frac{1}{\log(2v M e^{v}|R|^{v}|B|)}\cdot\frac{1}{v}\cdot\suc_{}^{\rm{NS}}(\cN, M, \rho),\label{eq:event1}
\end{align}
where we used $B^{j_m} \simeq R$ and  in the  inequality we used the pinching inequality \eqref{eq:pinching-ineq}
\begin{equation}
  v \cdot \sum_{j=1}^{v} \Pi^{j}_R \Lambda_{R B} \Pi^{j}_R = v \cdot \cP_{\rho}(\Lambda_{R B}) \mge \Lambda_{R B}.  
\end{equation}
To bound the expectation  \eqref{eq:rdm2}, we use concentration of random matrices as  in \cite{Fawzi19}.  Since $\Lambda_{RB}$ is a  solution to the program 
    \eqref{eq:ns-prog-q-proof}, we have $ \Lambda_{RB} \mle \rho_R \otimes  \mathbb{I}_B$ so $ \Pi_{B^{j_m}}^{j_m} \Lambda_{B^{j_m} B} \Pi_{B^{j_m}}^{j_m}\mle p_{j_m} \mathbf{1}_{B^{j_m}_{j_m}}\otimes \dI_B $ thus
\begin{align}
    &\exs{\bm{j}\sim p, \bm{U}\sim \cL}{\frac{1}{p_{j_m}}\cdot\frac{1}{\log(2v M e^{v}|R|^{v}|B|)}\cdot \tr{ \Pi_{B^{j_m}}^{j_m} \Lambda_{B^{j_m} B} \Pi_{B^{j_m}}^{j_m} \cdot(J_{\cN})_{B^{j_m} B} }\cdot \bid(\cG^c)}
    \\&\le \frac{1}{\log(2v M e^{v}|R|^{v}|B|)}\cdot\exs{\bm{j}\sim p, \bm{U}\sim \cL}{\tr{ \mathbf{1}_{B^{j_m}_{j_m}}\otimes \dI_B \cdot(J_{\cN})_{B^{j_m} B} }\cdot \bid(\cG^c)}
     \\&= \frac{1}{\log(2v M e^{v}|R|^{v}|B|)}\cdot\exs{\bm{j}\sim p, \bm{U}\sim \cL}{\tr{ \mathbf{1}_{B^{j_m}_{j_m}} \cdot \dI_{B^{j_m} } }\cdot \bid(\cG^c)}
    \\&= \frac{1}{\log(2v M e^{v}|R|^{v}|B|)}\cdot
    \mathbb{P}_{\bm{j}\sim p, \bm{U}\sim \cL} \left[\left\|\textstyle\sum_{m=1}^M O^{\bm{j}, \, \bm{U}}(m)\right\|_{\infty} >\log(2v M e^{v}|R|^{v}|B|)\right],\label{eq:event2}
\end{align}
where we recall the definition of the random matrix $O^{\bm{j}, \, \bm{U}}_{B^1 \cdots B^{v}B}(m)$:
\begin{align}
   O^{\bm{j}, \, \bm{U}}_{B^1 \cdots B^{v}B}(m) = \frac{|B_{j_m}^{j_m}|}{p_{j_m}}\cdot (U_m)_{B_{j_m}}\Pi_{B^{j_m}}^{j_m} \Lambda_{B^{j_m} B} \Pi_{B^{j_m}}^{j_m}(U_m^\dagger)_{B_{j_m}} \otimes \bigotimes_{s\in [v]\setminus j_m} \dI_{B^s}.
\end{align}
Moreover, $ \Pi_{B^{j_m}}^{j_m} \Lambda_{B^{j_m} B} \Pi_{B^{j_m}}^{j_m}\mle p_{j_m} \mathbf{1}_{B_{j_m}^{j_m}}\otimes \dI_B $ implies
\begin{align}
    0\mle O^{\bm{j}, \, \bm{U}}_{B^1 \cdots B^{v}B}(m) &= \tfrac{|B_{j_m}^{j_m}|}{p_{j_m}}\cdot (U_m)_{B_{j_m}}\Pi_{B^{j_m}}^{j_m} \Lambda_{B^{j_m} B} \Pi_{B^{j_m}}^{j_m}(U_m^\dagger)_{B_{j_m}} \otimes \bigotimes_{s\in [v]\setminus j_m} \dI_{B^s}
    \\&\mle |B_{j_m}^{j_m}|\cdot (U_m)_{B^{j_m}_{j_m}}\mathbf{1}_{B^{j_m}_{j_m}}\otimes \dI_B (U_m^\dagger)_{B^{j_m}_{j_m}} \otimes \bigotimes_{s\in [v]\setminus j_m} \dI_{B^s}
    \\&= \dI_{B_{j_m}^{j_m}}\otimes \dI_B\otimes \bigotimes_{s\in [v]\setminus j_m} \dI_{B^s}
    \\&\mle \dI_{B^{j_m}}\otimes \bigotimes_{s\in [v]\setminus j_m} \dI_{B^s}\otimes \dI_B
    \\&= \dI_{B^1\cdots B^{v}}\otimes \dI_B\,,
\end{align}
and thus $0\le\lambda_{\min}(O^{\bm{j}, \, \bm{U}}(m))$ as well as $\lambda_{\max}(O^{\bm{j}, \, \bm{U}}(m))\le 1 =L$.  Furthermore,  the expectation of the random matrix $O^{\bm{j}, \, \bm{U}}(m)$ can be  controlled  as follows:
\begin{align}
   & \exs{\bm{j}\sim p, \bm{U}\sim \cL}{O^{\bm{j}, \, \bm{U}}_{B^1 \cdots B^{v}B}(m)}
   \\&= \sum_{j_m=1}^v p_{j_m} \exs{U_m \sim \cL(B_{j_m}^{j_m})}{ \tfrac{ |B^{j_m}_{j_m}|}{p_{j_m}}\cdot (U_m)_{B_{j_m}^{j_m}}\Pi_{B^{j_m}}^{j_m} \Lambda_{B^{j_m} B} \Pi_{B^{j_m}}^{j_m}(U_m^\dagger)_{B_{j_m}^{j_m}}\otimes \bigotimes_{s\in [v]\setminus j_m} \dI_{B^s}}
   \\&= \sum_{j_m=1}^v \mathbf{1}_{B^{j_m}_{j_m}}\otimes\ptr{B^{j_m}_{j_m}}{ {|B^{j_m}_{j_m}|}\cdot \Pi_{B^{j_m}}^{j_m} \Lambda_{B^{j_m} B} \Pi_{B^{j_m}}^{j_m}\otimes \bigotimes_{s\in [v]\setminus j_m} \dI_{B^s}}
    \\&\mle  \sum_{j_m=1}^v  \dI_{B^{j_m}_{j_m}} \otimes \ptr{B^{j_m}}{\Lambda_{B^{j_m} B}\otimes \bigotimes_{s\in [v]\setminus j_m} \dI_{B^s}}
    \\&\mle   \sum_{j_m=1}^v  \dI_{B^{j_m}} \otimes \Lambda_{B} \otimes \bigotimes_{s\in [v]\setminus j_m} \dI_{B^s}
    \\&= \frac{v}{M}\dI_{B^1\cdots B^{v}}\otimes \dI_B.
\end{align}
Hence, we deduce that 
$\mu_{\max} = \lambda_{\max}\left(\sum_{m=1}^M \exs{\bm{j}\sim p, \bm{U}\sim \cL}{O^{\bm{j}, \, \bm{U}}_{B^1 \cdots B^{v}B}(m)}\right)\le \lambda_{\max}\left(v\dI_{B^1\cdots B^{v}}\otimes \dI_B\right)= v$. 
Therefore by the matrix Chernoff inequality ($L=1$ and $\mu_{\max}\le v$) (see \cite[Theorem 5.1.1]{Tropp2015May} reproduced in   Theorem \ref{thm:rdm-matrix}) we have for $ \gamma = (1+\delta)\mu_{\max} =  \log(2v M |R|^{v}|B|)$ that
\begin{align}
    \mathbb{P}_{\bm{j}\sim p, \bm{U}\sim \cL} \left[\big\|\textstyle\sum_{m=1}^M O^{\bm{j}, \, \bm{U}}(m)\big\|_{\infty} >(1+\delta)\mu_{\max} \right]&\le |B^1 \cdots B^{v}B| \cdot \left(\frac{e^{\delta}}{(1+\delta)^{1+\delta}}\right)^{\mu_{\max}}
    \\&\overset{(a)}{\le} |R|^{v}|B| \exp(-\gamma)
      \\&\overset{(b)}{\le} \frac{1}{2}\cdot\frac{1}{v}\cdot  \suc^{\rm{NS}}(\cN, M,\rho),\label{eq:event3}
\end{align}
where we used in $(a)$ $\mu_{\max}\le v$ and  the inequality $\gamma\log \gamma -\gamma -\gamma\ln(v) \ge  \gamma$ valid for all $\gamma\ge e^2v$ (this is implied by the choice of $\gamma$ and the assumption  $\log|A|\ge e^2$); in $(b)$ we used $\suc_{}^{\rm{NS}}(\cN, M, \rho)  \ge \frac{1}{M}$ and we choose $\gamma$ such that 
\begin{align}
    &|R|^{v}|B| \exp(-\gamma)\le \frac{1}{2}\cdot\frac{1}{v}\cdot \frac{1}{M}
   \Leftrightarrow \gamma \ge  \log(2v M |R|^{v}|B|).
\end{align}
Consequently, from \eqref{eq:event1},  \eqref{eq:event2}, and \eqref{eq:event3}, we have that
\begin{align}
    &\exs{\bm{j}\sim p, \bm{U}\sim \cL}{\suc(\cN,M, m)}
    \\&\ge \frac{1}{v\log(2v M e^{v}|R|^{v}|B|)}\cdot  \suc^{\rm{NS}}(\cN, M,\rho) 
   - \frac{1}{2v\log(2v M e^{v}|R|^{v}|B|)}\cdot \suc^{\rm{NS}}(\cN, M,  \rho) 
    \\&= \frac{1}{2v \log(2v M e^{v}|R|^{v}|B|)}\cdot \suc_{}^{\rm{NS}}(\cN,M,\rho).
\end{align}
Finally since $R \simeq A$ we conclude that:
\begin{align}
    \suc^{\rm{EA}}(\cN,M) &\ge \frac{1}{M}\sum_{m=1}^M  \exs{\bm{j}\sim p, \bm{U}\sim \cL}{\suc(\cN,M, m)}
    \\&\ge \frac{1}{2v \log(2v M e^{v}|A|^{v}|B|)}\cdot\suc^{\rm{NS}}(\cN, M, \rho) .
\end{align}
\end{proof}


\subsection{Application: strong converse exponent for coding over quantum channels}\label{sec:qq-application-sce}

In this section, we apply the approximation result with multiplicative error obtained in Section \ref{sec:qq-multiplicative} to establish the entanglement-assisted strong converse exponent in the quantum setting. 
To this end, we show that the coding strong converse exponent is the same for entanglement-assisted and non-signaling strategies.
\begin{prop}\label{prop:qq-EE-EA-NS}
    Let $\cN$ be a quantum channel. For all $r\ge 0$, we have  that
    \begin{align}
   E^{\rm{EA}}(\cN, r) = E^{\rm{NS}}(\cN, r).
\end{align}
\end{prop}
As a direct corollary of this proposition and the exact characterization of the non-signaling strong converse exponent in Corollary \ref{cor:SCE-MC-NS}, we can conclude the entanglement-assisted strong converse exponent.

\begin{cor}\label{cor:SCE-EA-NS}
    Let  $\cN$ be a quantum channel. For all $r\ge 0$, we have  that
 \begin{align}
     E^{\rm{EA}}(\cN, r) =  \sup_{\alpha\ge 0}  \frac{\alpha}{1+\alpha}\left( r- \widetilde{\cI}_{1+\alpha}(\cN) \right). 
    \end{align}
\end{cor}

We note that \cite{Li2023Nov} is the first to prove this result. Our proof strategy differs significantly and may be of independent interest.

\begin{proof}[Proof of Proposition \ref{prop:qq-EE-EA-NS}] We first  reduce the problem to a permutation invariant state $\rho_R$ in the non-signaling program of $\suc^{\rm{NS}}(\cN^{\otimes n},M, \rho_{R^n})$ \eqref{eq:ns-program-rho}. Next,  we reduce the problem to the case where $\rho_{R^n} = \rho_{R^n}^{\rm{dF}}$ is the de Finetti state \cite{christandl2009postselection}. Finally we apply the rounding result proven in Proposition \ref{prop:qq-rounding}.\\

\textbf{Reduction to permutation invariant states.} 
By Lemma~\ref{lem:red-perm-invt-states}, the optimization in 
\begin{align}
    \suc^{\rm{NS}}(\cN^{\otimes n}, M) = \sup_{\rho_{R^n}, \Lambda_{R^n B^n}} \quad & \tr{\Lambda_{R^n B^n} \cdot (J_{\cN})_{RB}^{\otimes n}} \\
    \text{subject to} \quad
    & \rho_{R^n} \in \cS(R^n), \\
    & \Lambda_{B^n} = \tfrac{1}{M} \mathbb{I}_{B^n}, \\
    & 0 \mle \Lambda_{R^n B^n} \mle \rho_{R^n} \otimes \mathbb{I}_{B^n}.
\end{align}

can be restricted to permutation invariant states $\rho_{R^n}$.  We denote the set of permutation invariant states over $R^{\otimes n}$ by $\cS_n(R)$. We note that similar symmetry arguments were used by \cite{Matthews2014Sep,Leung2015Jun} to reduce the NS programs of the coding size and fidelity.\\

\textbf{Reduction to the de Finetti state.} Let  $g(n,d) = \binom{n+d^2-1}{n}\le (n+1)^{d^2-1}$. Any permutation invariant state $\rho_{R^n}\in \cS_n(R)$ satisfies $ \rho_{R^n} \mle g(n,|R|) \rho_{R^n}^{\rm{dF}}$ where $\rho_{R^n}^{\rm{dF}}$ is the de Finetti state from \cite{christandl2009postselection} and stated later on in \eqref{eq:definetti}. So, we find
\begin{align}
    &\suc_{}^{\rm{NS}}(\cN^{\otimes n},M) \\&=\sup_{\rho_{R^n} \in \cS_n(R)} \sup_{\Lambda_{R^nB^n}} \left\{\tr{\Lambda_{R^nB^n} \cdot (J_{\cN})_{RB}^{\otimes n}} \;\middle|\; \Lambda_{B^n} = \tfrac{1}{M}\mathbb{I}_{B^n},\, 0\mle  \Lambda_{R^nB^n} \mle \rho_{R^n} \otimes  \mathbb{I}_{B^n}\right\}
           \\&\le  \sup_{\Lambda_{R^nB^n}} \left\{\tr{\Lambda_{R^nB^n} \cdot (J_{\cN})_{RB}^{\otimes n}} \; \middle| \; \Lambda_{B^n} = \tfrac{1}{M}\mathbb{I}_{B^n}, \,0\mle  \Lambda_{R^nB^n} \mle g(n,|R|) \rho_{R^n}^{\rm{dF}} \otimes  \mathbb{I}_{B^n}   \right\}
           \\&= \sup_{\Lambda_{R^nB^n}} \left\{g(n,|R|)\tr{\Lambda_{R^nB^n}\cdot (J_{\cN})_{RB}^{\otimes n}} \; \middle| \; \Lambda_{B^n} = \tfrac{1}{g(n,|R|) M}\mathbb{I}_{B^n}, \,0\mle  \Lambda_{R^nB^n} \mle  \rho_{R^n}^{\rm{dF}} \otimes  \mathbb{I}_{B^n}   \right\}
           \\&= g(n,|R|)\cdot \suc_{}^{\rm{NS}}(\cN^{\otimes n}, g(n,|R|)M ,  \rho^{\rm{dF}}) 
           \\&\le g(n,|R|) \cdot \sup_{\rho_{R^n} \in \cS(R^n)}  \suc_{}^{\rm{NS}}(\cN^{\otimes n}, g(n,|R|)M ,  \rho) 
          \\&\le g(n,|R|)\cdot \suc_{}^{\rm{NS}}(\cN^{\otimes n}, g(n,|R|) M).
\end{align}
Hence, the non-signaling success probability  $\suc_{}^{\rm{NS}}(\cN^{\otimes n},M)$ is comparable   to  $\suc_{}^{\rm{NS}}(\cN^{\otimes n}, M, \rho^{\rm{dF}})$:
\begin{align}
     \tfrac{1}{g(n,|R|)}\suc^{\rm{NS}}(\cN^{\otimes n},M)&\le \tfrac{1}{g(n,|R|)}\suc^{\rm{NS}}(\cN^{\otimes n}, \tfrac{M}{g(n,|R|)}) 
     \\&\le \suc_{}^{\rm{NS}}(\cN^{\otimes n},M,  \rho^{\rm{dF}}) \le  \suc^{\rm{NS}}(\cN^{\otimes n},M).\label{eq:tilde-succ}
\end{align}
In particular to round the NS success probability it suffices to consider $\rho_R = \rho_R^{\rm{dF}}$ in the NS program \eqref{NS-program} albeit with a rounding loss of $\frac{1}{g(n,|R|)}$.\\

\textbf{Application of the rounding protocol of Proposition \ref{prop:qq-rounding}. } 
Note that the de Finetti state can be written as \cite[Lemma 1]{Hayashi2016Oct}:
\begin{align}\label{eq:definetti}
    \rho^{\rm{dF}}_{R_1\cdots R_n} = \bigoplus_{\lambda\in Y_{n,|R|}} p_{\lambda} \frac{\Pi_{\lambda}}{\tr{\Pi_{\lambda}}},
\end{align}
where $Y_{n,|R|}$ is the set of Young diagrams of size $n$ and depth $|R|$ that satisfies $|Y_{n,|R|}|\le (n+1)^{|R|-1}$, $\{\Pi_{\lambda}\}_{\lambda\in Y_{n,|R|}}$ are orthogonal projectors and $\{p_{\lambda}\}_{\lambda\in Y_{n,|R|}}$ is a probability distribution. We can then apply Proposition \ref{prop:qq-rounding}  for $\cN^{\otimes n}$, $M=e^{nr}$,  $\rho_R= \rho_R^{\rm{dF}}$ and $v= |Y_{n,|R|}|$ to obtain:
    \begin{align}
         &\suc^{\rm{EA}}(\cN^{\otimes n},e^{nr}) 
         \\&\ge \frac{1}{2|Y_{n,|R|}| \log(2|Y_{n,|R|}| e^{nr} e^{|Y_{n,|R|}|}|A|^{n|Y_{n,|R|}|}|B|^n)}\cdot\suc^{\rm{NS}}(\cN^{\otimes n},e^{nr}, \rho^{\rm{dF}})
         \\&\ge \frac{1}{2g(n,|A|)|Y_{n,|A|}|\log(2|Y_{n,|A|}| e^{nr} e^{|Y_{n,|A|}|}|A|^{n|Y_{n,|A|}|}|B|^n)}\cdot\suc^{\rm{NS}}(\cN^{\otimes n},e^{nr}),
    \end{align}
where we used \eqref{eq:tilde-succ} and $R\simeq A$ in the last inequality. Note that using $g(n,|A|) = \binom{n+|A|^2-1}{n}\le (n+1)^{|A|^2-1}$ and $|Y_{n,|A|}|\le (n+1)^{|A|-1}$ we have 
\begin{align}
    &\frac{1}{n}\log \left[2g(n,|A|)|Y_{n,|A|}|\log\left(2|Y_{n,|A|}| e^{nr}e^{|Y_{n,|A|}|} |A|^{n|Y_{n,|A|}|}|B|^n\right) \right]
    \\&\le  \frac{\log2}{n}+ (|A|^2+|A|-2)\frac{\log(n+1)}{n}
    \\&\quad + \frac{1}{n}\log\left[\log(2)+ (|A|-1)\log(n+1) +nr +(n+1)^{|A|} \log(e|A|) + n \log(|B|) \right] 
    \\&\le  \frac{\log2}{n}+ (|A|^2+|A|-2)\frac{\log(n+1)}{n}+ \frac{1}{n}\log[O((n+1)^{|A|}) ] 
    \\&\le  \frac{\log2}{n}+ (|A|^2+|A|-2)\frac{\log(n+1)}{n}+ O(|A|)\frac{\log(n+1)}{n} 
    \underset{n \rightarrow \infty}{\longrightarrow} 0. 
\end{align}
Hence, we find
\begin{align}
        \suc^{\rm{NS}}(\cN^{\otimes n},e^{nr})\ge  \suc^{\rm{EA}}(\cN^{\otimes n},e^{nr}) &\ge e^{-o(n)}\cdot\suc^{\rm{NS}}(\cN^{\otimes n},e^{nr}).
    \end{align}
Finally, we get
\begin{align}
    E^{\rm{EA}}(\cN, r) &= \lim_{n \rightarrow \infty} - \frac{1}{n}\log\suc^{\rm{EA}}(\cN^{\otimes n},e^{nr})
    \\&= \lim_{n \rightarrow \infty} -\frac{1}{n}\log\suc^{\rm{NS}}(\cN^{\otimes n},e^{nr})=  E^{\rm{NS}}(\cN, r).
\end{align}
\end{proof}


\section{Conclusion}

We explored entanglement-assisted quantum channel coding from an algorithmic point of view. We developed approximation algorithms for the entanglement-assisted success probability using non-signaling relaxation in the quantum-classical and in the quantum setting. As an application, we provided an alternative proof for the characterization of the entanglement-assisted strong converse exponent. The achievability proof proceeds by connecting the entanglement-assisted (EA), non-signaling (NS), and meta-converse (MC) strong converse exponents. The latter is then deduced from the inherent relation between the meta-converse and composite hypothesis testing. While this technique worked effectively for strong converse exponents, it remains unclear whether the EA, NS, and MC error exponents are equal. Moreover, obtaining tight approximation bounds in the fully quantum setting remains an open problem.


\section*{Acknowledgments}

MB and AO acknowledge funding by the European Research Council (ERC Grant Agreement No. 948139), MB acknowledges support from the Excellence Cluster - Matter and Light for Quantum Computing (ML4Q).

\printbibliography


\appendix

\section{Dual formulation of the meta-converse}\label{app:dual-MC}

\begin{prop} \label{prop-dual-MC}
Let  $\cN$ be a  quantum channel. The optimal one-shot  meta-converse  success probability satisfies: 
\begin{align}
   \suc_{}^{\rm{MC}}(\cN,M)  &= \sup_{\rho_R \in \cS(R)}\inf_{Z_B\mge 0}  \tr{(\rho_R^{1/2} ( J_{\cN})_{RB}\rho_R^{1/2} - M\rho_R\otimes Z_B)_+} + \tr{Z_B}
   \\&= \sup_{{\rho}_{R} \in \cS(R)}\sup_{0\mle O\mle \dI} \left\{\tr{\rho_R^{1/2} ( J_{\cN})_{RB}\rho_R^{1/2} \cdot O} \; \middle| \;  \sup_{\sigma_B\in \cS(B)}\tr{(\rho_R\otimes \sigma_B) O}\le \tfrac{1}{M} \right\} \label{eq:MC-dual2}
     \\&= \sup_{{\rho}_{R} \in \cS(R)} \inf_{\sigma_B\in \cS(B)}\sup_{0\mle O\mle \dI} \left\{\tr{\rho_R^{1/2} ( J_{\cN})_{RB}\rho_R^{1/2} \cdot O} \; \middle| \; \tr{(\rho_R\otimes \sigma_B) O}\le \tfrac{1}{M} \right\}. \label{eq:MC-dual3}
\end{align}
\end{prop}

\begin{proof}
The first equality was proven in \cite[Proposition 6.1]{Oufkir24}.
Let us  prove \eqref{eq:MC-dual2}. We have that 
\begin{align}
  \suc^{\rm{MC}}(\cN,M) &=  \sup_{\rho_R\in \cS(R)}\inf_{Z_B\mge 0} \tr{\big(\rho_R^{1/2} ( J_{\cN})_{RB}\rho_R^{1/2} - M\rho_R\otimes Z_B\big)_+}+\tr{Z_B}
  \\&=\sup_{\rho_R\in \cS(R)}\inf_{Z_B \mge 0} \sup_{0\mle O\mle \dI} \tr{\big(\rho_R^{1/2} ( J_{\cN})_{RB}\rho_R^{1/2} -  M\rho_R\otimes Z_B\big)\cdot O} + \tr{Z_B}
   \\&=\sup_{\rho_R\in \cS(R)} \sup_{0\mle O\mle \dI} \inf_{Z_B \mge 0}\tr{\big(\rho_R^{1/2} ( J_{\cN})_{RB}\rho_R^{1/2} -  M\rho_R\otimes Z_B\big)\cdot O} + \tr{Z_B},
\end{align}
 where we used Sion's minimax theorem (see Lemma \ref{lem-Sion}) in the last equality. By writing $Z_B\mge 0 $ as  $Z_B= z\sigma_B$ where $z=\tr{Z}\ge 0$ and $\sigma_B \in \cS(B)$ we have 
\begin{align}
  &\suc^{\rm{MC}}(\cN,M) 
  \\&=  \sup_{\rho_R\in \cS(R)} \sup_{0\mle O\mle \dI} \inf_{Z_B \mge 0}\tr{\big(\rho_R^{1/2} ( J_{\cN})_{RB}\rho_R^{1/2} -  M\rho_R\otimes Z_B\big)\cdot O} + \tr{Z_B}
  \\&=\sup_{\rho_R\in \cS(R)}\sup_{0\mle O\mle \dI} \inf_{z\mge 0}\tr{(\rho_R^{1/2} ( J_{\cN})_{RB}\rho_R^{1/2})\cdot O }+z\left(1- M\sup_{\sigma_B \in \cS(B)}\tr{(\rho_R\otimes \sigma_B)\cdot O}\right)
     \\& =\sup_{\rho_R\in \cS(R)}\sup_{\substack{0\mle O\mle \dI \\ \sup_{\sigma_B \in \cS(B)}\tr{(\rho_R\otimes \sigma_B)\cdot O}\le \frac{1}{M}}} \inf_{z\mge 0}\tr{(\rho_R^{1/2} ( J_{\cN})_{RB}\rho_R^{1/2})\cdot O }
     \\&\qquad\qquad\qquad\qquad\qquad\qquad\qquad\qquad\qquad +z\left(1- M\sup_{\sigma_B \in \cS(B)}\tr{(\rho_R\otimes \sigma_B)\cdot O}\right)
       \\&=\sup_{\rho_R\in \cS(R)}\sup_{\substack{0\mle O\mle \dI \\ \sup_{\sigma_B \in \cS(B)}\tr{(\rho_R\otimes \sigma_B)\cdot O}\le \frac{1}{M}}} \tr{(\rho_R^{1/2} ( J_{\cN})_{RB}\rho_R^{1/2})\cdot O }.
\end{align}
 The proof of \eqref{eq:MC-dual3} is similar: 
\begin{align}
  &\suc^{\rm{MC}}(\cN,M) 
  \\&=  \sup_{\rho_R\in \cS(R)}\sup_{0\mle O\mle \dI} \inf_{\sigma_B \in \cS(B)}\inf_{z\mge 0}\tr{(\rho_R^{1/2} ( J_{\cN})_{RB}\rho_R^{1/2})\cdot O }+z\left(1- M\tr{(\rho_R\otimes \sigma_B)\cdot O}\right)
     \\& =\sup_{\rho_R\in \cS(R)}\inf_{\sigma_B \in \cS(B)}\sup_{\substack{0\mle O\mle \dI \\ \tr{(\rho_R\otimes \sigma_B)\cdot O}\le \frac{1}{M}}} \inf_{z\mge 0}\tr{(\rho_R^{1/2} ( J_{\cN})_{RB}\rho_R^{1/2})\cdot O }
       \\&\qquad\qquad\qquad\qquad\qquad\qquad\qquad\qquad\;  +z\left(1- M\tr{(\rho_R\otimes \sigma_B)\cdot O}\right)
       \\&=\sup_{\rho_R\in \cS(R)}\inf_{\sigma_B \in \cS(B)}\sup_{\substack{0\mle O\mle \dI \\ \tr{(\rho_R\otimes \sigma_B)\cdot O}\le \frac{1}{M}}} \tr{(\rho_R^{1/2} ( J_{\cN})_{RB}\rho_R^{1/2})\cdot O }.
\end{align}
\end{proof}


\section{Technical statements}

\begin{lem}[\cite{sion1958general}]\label{lem-Sion}
    Let $X$ be a compact convex subset of a linear topological space and $Y$ a convex subset of a linear topological space. If $f$  is a real-valued function on $X \times Y$ with
    \begin{itemize}
        \item $f(x, \cdot)$ is upper semi-continuous and quasi-concave on $Y$, $ \forall x\in X $, and
        \item $f(\cdot,y)$ is lower semi-continuous and quasi-convex on $X$, $ \forall y\in Y $
    \end{itemize}
    then
\begin{equation}
    \min_{x\in X} \sup_{y\in Y} f(x,y) =  \sup_{y\in Y} \min_{x\in X}f(x,y).
\end{equation}
\end{lem}

\begin{lem}\label{lem:red-perm-invt-states}
Let $n, M\in \mathbb{N}$ and $\cN$ be a quantum channel.
The optimization in 
  \begin{align}
         &  \suc^{\rm{NS}}(\cN^{\otimes n},M)
         \\&= \sup_{\rho_{R^n} \in \cS(R^n)} \sup_{\Lambda_{R^nB^n}} \left\{\tr{\Lambda_{R^nB^n} \cdot (J_{\cN})_{RB}^{\otimes n}} \;\middle|\; \Lambda_{B^n} = \tfrac{1}{M}\mathbb{I}_{B^n},\, 0\mle  \Lambda_{R^nB^n} \mle \rho_{R^n} \otimes  \mathbb{I}_{B^n}\right\}
\end{align}
can be restricted to permutation invariant states $\rho_{R^n}$. 
\end{lem}
\begin{proof}
Define the linear objective function $f(\Lambda_{RB}, \rho_{R})=\tr{\Lambda_{RB} \cdot (J_{\cN})_{RB}} $ and the set \[\cC=\{(\Lambda_{RB}, \rho_{R})|\Lambda_B = \tfrac{1}{M}\mathbb{I}_B, \;  0\mle  \Lambda_{RB} \mle \rho_R \otimes  \mathbb{I}_B, \; \rho_R \in \cS(R)\}.\] 
$\cC$ is a convex set, indeed take $(\Lambda_{RB}, \rho_{R})\in \cC$ and $(\Gamma_{RB}, \sigma_{R})\in \cC$ and $\lambda\in [0,1]$ we have that 
\begin{itemize}
    \item $(\lambda \Lambda +(1-\lambda)\Gamma)_B = \lambda \Lambda_{B} +(1-\lambda)\Gamma_{B} = \lambda \frac{1}{M}\dI_{B} +(1-\lambda)\frac{1}{M}\dI_{B}=\frac{1}{M}\dI_{B}$,
    \item $0\mle \lambda \Lambda_{RB} +(1-\lambda)\Gamma_{RB} \mle  \lambda \rho_{R}\otimes \dI_B +(1-\lambda)\sigma_{R}\otimes \dI_B = (\lambda \rho_{R} +(1-\lambda)\sigma_{R})\otimes \dI_B $,
    \item $\lambda \rho_{R} +(1-\lambda)\sigma_{R}\in \cS(R)$
\end{itemize}
so $\lambda(\Lambda_{RB}, \rho_{R})+(1-\lambda)(\Gamma_{RB}, \sigma_{R})\in \cC$. 

Let $n\in \mathbb{N}$ and   $\pi \in \fS_n$. Let $U^\pi_{\cH^n}$ be the corresponding permutation unitary on $\cH^{\otimes n}$. With $n$ uses of the channel $\cN$, define $f_n(\Lambda_{R^nB^n}, \rho_{R^n})=\tr{\Lambda_{R^nB^n} \cdot (J_{\cN})_{RB}^{\otimes n}}$. We have that 
\begin{align}
    f_n(U^\pi_{R^nB^n} \Lambda_{R^nB^n} (U^\pi_{R^nB^n})^{\dagger}, U^\pi_{R^n} \rho_{R^n} (U^\pi_{R^n})^{\dagger} ) &= \tr{U^\pi_{R^nB^n}\Lambda_{R^nB^n}(U^\pi_{R^nB^n})^{\dagger} \cdot (J_{\cN})_{RB}^{\otimes n}}
    \\&=\tr{\Lambda_{R^nB^n} \cdot(U^\pi_{R^nB^n})^{\dagger} (J_{\cN})_{RB}^{\otimes n}U^\pi_{R^nB^n}}
    \\&= \tr{\Lambda_{R^nB^n} \cdot (J_{\cN})_{RB}^{\otimes n}}
    \\&=f_n(\Lambda_{R^nB^n}, \rho_{R^n}).
\end{align}
Moreover $(U^\pi_{R^nB^n} \Lambda_{R^nB^n} (U^\pi_{R^nB^n})^{\dagger}, U^\pi_{R^n} \rho_{R^n} (U^\pi_{R^n})^{\dagger})\in \cC$ as 
\begin{itemize}
    \item $(U^\pi_{R^nB^n} \Lambda_{R^nB^n} (U^\pi_{R^nB^n})^{\dagger})_{B^n}  = U^\pi_{B^n} \Lambda_{B^n}  (U^\pi_B)^{\dagger} = U^\pi_{B^n} \frac{1}{M}\dI_{B^n} (U^\pi_B)^{\dagger} = \frac{1}{M}\dI_{B^n} $,
    \item $0\mle U^\pi_{R^nB^n} \Lambda_{R^nB^n} (U^\pi_{R^nB^n})^{\dagger} \mle U^\pi_{R^n}U^\pi_{B^n} (\rho_{R^n} \otimes \dI_{B^n} ) (U^\pi_{R^n})^{\dagger}(U^\pi_{B^n})^{\dagger} = U^\pi_{R^n} \rho_{R^n} (U^\pi_{R^n})^{\dagger}\otimes \dI_{B^n}$, 
    \item $U^\pi_{R^n} \rho_{R^n} (U^\pi_{R^n})^{\dagger}\in \cS(R)$.
\end{itemize}
Hence,  the function $f_n$ achieves its supremum over $\cC$ at a permutation invariant pair $(\Lambda_{R^nB^n}, \rho_{R^n})$. 
\end{proof}


\section{Converse result for the meta-converse strong converse}\label{app:proof-converse}

Here, we prove Proposition \ref{prop:SCE-converse} which we restate.

\begin{prop*}[Restatement of Proposition \ref{prop:SCE-converse}]
     Let  $\cN$ be a quantum channel. For all $r\ge 0$, we have  that
 \begin{align}
     E^{\rm{MC}}(\cN, r) = \lim_{n\rightarrow \infty } - \frac{1}{n}\log \suc^{\rm{MC}}(\cN^{\otimes n}, e^{nr})  \ge \sup_{\alpha\ge 0}  \frac{\alpha}{1+\alpha}\left( r- \widetilde{\cI}_{1+\alpha}(\cN) \right). 
    \end{align}
\end{prop*}

\begin{proof}
    From Proposition \ref{prop-dual-MC} we have that the meta-converse  success probability satisfies: 
\begin{align}
   \suc_{}^{\rm{MC}}(\cN,M)  &= \sup_{{\rho}_{R} \in \cS(R)} \inf_{\sigma_B\in \cS(B)}\sup_{0\mle O\mle \dI} \left\{\tr{\rho_R^{1/2} ( J_{\cN})_{RB}\rho_R^{1/2} \cdot O} \; \middle| \; \tr{(\rho_R\otimes \sigma_B) O}\le \tfrac{1}{M} \right\}. 
\end{align}
Denote by $\rho_{RB} = \rho_R^{1/2} ( J_{\cN})_{RB}\rho_R^{1/2}$. Let $\alpha\ge 0$ and  $0\mle O\mle \dI $ be an observable such that $\tr{(\rho_R\otimes \sigma_B) O}\le \tfrac{1}{M}$. By the data processing property of the sandwiched R\'enyi divergence applied with the measurement channel $\cK(\cdot)= \tr{(\cdot)O}\proj{0}+ \tr{(\cdot)(\dI-O)}\proj{1}$:
\begin{align}
    \widetilde{D}_{1+\alpha}(\rho_{RB}\|\rho_R\otimes \sigma_B) &\ge \widetilde{D}_{1+\alpha}(\cK(\rho_{RB})\|\cK(\rho_R\otimes \sigma_B))
    \\&\ge \frac{1}{\alpha}\log\left( \tr{\rho_{RB}O}^{1+\alpha}\tr{(\rho_R\otimes \sigma_B) O}^{-\alpha} \right)
    \\&\ge \frac{1+\alpha}{\alpha}\log( \tr{\rho_{RB}O})+\log(M),\label{eq:sand}
\end{align}
where we used $\tr{(\rho_R\otimes \sigma_B) O}\le \tfrac{1}{M}$ in the last inequality. Hence, for all $\alpha\ge 0$
\begin{align}
   &\suc_{}^{\rm{MC}}(\cN,M)  \\&= \sup_{{\rho}_{R} \in \cS(R)} \inf_{\sigma_B\in \cS(B)}\sup_{0\mle O\mle \dI} \left\{\tr{\rho_R^{1/2} ( J_{\cN})_{RB}\rho_R^{1/2} \cdot O} \; \middle| \; \tr{(\rho_R\otimes \sigma_B) O}\le \tfrac{1}{M} \right\}
    \\ &\le  \sup_{{\rho}_{R} \in \cS(R)} \inf_{\sigma_B\in \cS(B)}  \exp\left(-\frac{\alpha}{1+\alpha} \left[ \log(M) - \widetilde{D}_{1+\alpha}\left(\rho_R^{1/2} ( J_{\cN})_{RB}\rho_R^{1/2}\big\|\rho_R\otimes \sigma_B\right) \right]\right)
       \\ &=   \exp\left(-\frac{\alpha}{1+\alpha} \left[ \log(M) - \widetilde{I}_{1+\alpha}\left(\cN\right) \right]\right).
\end{align}
Therefore 
\begin{align}
   \suc_{}^{\rm{MC}}(\cN^{\otimes n},e^{nr})  &\le   \inf_{\alpha\ge 0}  \exp\left(-\frac{\alpha}{1+\alpha} \left[ nr - \widetilde{I}_{1+\alpha}\left(\cN^{\otimes n}\right) \right]\right)
 \\&=   \inf_{\alpha\ge 0}   \exp\left(-n\frac{\alpha}{1+\alpha} \left[ r - \widetilde{I}_{1+\alpha}\left(\cN\right) \right]\right),
\end{align}
where we used the additivity of the sandwiched channel mutual information $\widetilde{I}_{1+\alpha}\left(\cN^{\otimes n}\right) = n\widetilde{I}_{1+\alpha}\left(\cN\right)$ \cite{Gupta2015Mar}. Finally 
\begin{align}
     - \frac{1}{n}\log \suc^{\rm{MC}}(\cN^{\otimes n}, e^{nr})  \ge \sup_{\alpha\ge 0}  \frac{\alpha}{1+\alpha}\left( r- \widetilde{\cI}_{1+\alpha}(\cN) \right) 
\end{align}
and the proposition follows by taking $n\rightarrow \infty$. 
\end{proof}


\section{Matrix Chernoff inequality}\label{sec:rdm-matrix}

The following theorem is part of  \cite[Theorem 5.1.1]{Tropp2015May}.

\begin{theorem}\label{thm:rdm-matrix}
    Consider a finite sequence $\{X_k\}$ of independent, random, Hermitian matrices with common dimension $d$. Assume that 
    \[
        0\le \lambda_{\min}(X_k) \quad \text{and} \quad \lambda_{\max}(X_k)\le L \quad \text{for each index }  k.\]
    Introduce the random matrix 
    \[        Y=\sum_k X_k.\]
Define the   maximum eigenvalue $\mu_{\max}$ of the expectation $\ex{Y}$:
\[\mu_{\max} = \lambda_{\max}(\ex{Y}) = \lambda_{\max}\left(\sum_k \ex{X_k} \right).\] 
Then 
\[\pr{\lambda_{\max}(Y)\ge (1+\delta)\mu_{\max}} \le d\left[\frac{e^{\delta}}{(1+\delta)^{1+\delta}}\right]^{\mu_{\max}/L}.\]
\end{theorem}


\end{document}